\documentclass{SCIS2023}

\usepackage{multirow}
\usepackage{color}
\usepackage{graphicx}

\usepackage[numbers,sort&compress]{natbib}

\begin{document}
\ArticleType{RESEARCH PAPER}
\Year{2022}
\Month{}
\Vol{}
\No{}
\DOI{}
\ArtNo{}
\ReceiveDate{}
\ReviseDate{}
\AcceptDate{}
\OnlineDate{}

\title{ The Theoretical Limit of Radar Target Detection}{Title keyword 5 for citation Title for citation Title for citation}

\author{Dazhuan XU}{{xudazhuan@nuaa.edu.cn}}
\author{Nan WANG}{}
\author{Han ZHAN}{}
\author{Xiaolong KONG}{}

\AuthorMark{Author A}

\AuthorCitation{Author A, Author B, Author C, et al}


\address{Electronic and information engineering, Nanjing University of Aeronautics and Astronautics, Nanjing {\rm 210000}, China}

\abstract{ In this paper, the problem of optimal radar target detection is investigated by applying Shannon information theory. The information-theoretical framework is established for a detection system model containing the target existence state, which consists of four parts. The definition of detection information (DI) is proposed to solve the quantitative problem of detection performance, which is a theoretical measure for evaluating the target detection system with the unit of bits. We prove that the \textit{a priori} probability of target existence is equal to the false alarm probability as long as the observation interval is large enough, establishing the relationship between the \textit{a priori} probability and the false alarm probability. The sampling \textit{a posteriori} (SAP) detection is proposed, whose empirical DI approaches the theoretical DI in detection performance. We prove that DI is the theoretical limit of target detection. Specifically, the empirical DI of SAP approaches the theoretical DI as the snapshots increase. Conversely, there is no detector whose empirical DI is greater than DI and the failure probability approaches zero. Numerical simulations are conducted to demonstrate that the empirical DI of the SAP detector approaches the DI with thousands of snapshots and that the empirical DI of MAP and NP detectors does less than DI.}


\maketitle
\section{Introduction}

Radar target detection has become a crucial component in various applications, encompassing aerospace\cite{1,2}, military surveillance\cite{3,4}, traffic monitoring\cite{5,6}, and environmental supervision\cite{7,8}. For instance, in military contexts, the capability to promptly detect targets is of paramount importance to safety and operational efficiency. Additionally, as the initial phase of radar operations, target detection performance holds significant implications for subsequent radar tasks such as target tracking, identification, or classification\cite{9,10,11}.

In radar target detection, the Neyman-Pearson (NP) criterion is commonly employed for designing detection systems. The NP criterion specifically refers to maximizing the detection probability under a given false alarm probability\cite{12,13}. Under the NP criterion, the generalized likelihood ratio test (GLRT) is adopted for detecting distributed targets, and the expression for the false alarm probability of the GLRT detector is derived\cite{14}. Reference \cite{15} presents the maximum a posteriori ratio test (MAPRT) detector and derives closed-form expressions for the false alarm probability and approximate expressions for the detection probability of the MAPRT detector for a dual-static radar system. To enhance the detection capability of distributed radars, reference \cite{16} proposes a constant false alarm rate (CFAR) detector and derives analytical expressions for the local false alarm probability and detection probability of the CFAR detector. Addressing range-Doppler extended targets, the adaptive subspace detector is proposed to improve detection performance in non-Gaussian clutter environments \cite{17}. For multiple-input multiple-output (MIMO) radar systems, reference \cite{18} states that the transmit waveform of MIMO radar should be matched to the type of target through the design of a generic target model.

However, the NP criterion applies only to "binary detection" scenarios, meaning either deciding that the target is present or absent. Binary detection is analogous to the hard decision in communication. The hard decision is capable of achieving optimality only in error rates but not in channel capacity, leading to its complete replacement by soft decision. This paper establishes the information theory for target detection by modeling target detection as a detection problem in an on-off keying (OOK) modulation system, pioneering an approach to solving the optimal target detection problem.

To quantify the detection information, we apply Shannon's information theory \cite{19,20,21} to radar target detection. The utilization of information theory in radar is traced back to the early 1950s. Woodward and Davies defined the information from radar observations as the entropy difference between the prior and posterior probability distributions of time delay \cite{22}. Additionally, Woodward provided a tutorial on the application of information theory to radar problems \cite{23}. Building upon Woodward's work, Bell designed optimal waveforms by maximizing the mutual information between the received signal and the target pulse response to enhance estimation performance \cite{24,25}. Since then, applying information theory in radar has attracted considerable attention. For example, the waveform design approach that maximizes Kullback-Leibler divergence is presented to enhance the detection capability of distributed targets under constraints on subcarrier power ratios \cite{26}. Reference \cite{27} states the relationships between three waveform design metrics: signal-to-noise ratio (SNR), Kullback-Leibler divergence, and mutual information.

We have conducted significant work in radar signal processing utilizing information theory. Reference \cite{28} defines joint range-scattering information and joint entropy error as forward and negative indicators, respectively, for evaluating estimation performance. The study proves that the theoretical joint range-scattering information and theoretical entropy error are the achievable theoretical limits of parameter estimation. The approximate closed-form expression for entropy error is derived in \cite{29,30}, indicating that entropy error serves as an effective performance bound. The novel definitions of joint resolution limit for the multi-parameter are proposed based on scattering information, and approximate closed-form expressions for the resolution limit are derived \cite{31,32,33}.

This paper presents the theoretical limit of radar target detection based on Shannon's information theory. We give a system model for detecting targets by introducing a target state $ v $ into a general radar system model, where $  v=0,1 $, for  $  H_{0} $ (signal absent), $  H_{1} $ (signal present), respectively. Taking into account the statistical properties of noise and target scattering, we derive the detection channel $ p\left( {{\boldsymbol{y}}\left| v \right.} \right) $, which is an equivalent communication system with on-off keying modulation. Then with the Bayesian formula, we derive the \textit{a posteriori} probability distribution function (PDF) $ P\left( {v\left| {\boldsymbol{y}} \right.} \right) $ of the target state. The main contributions of this paper are as follows:

\begin{itemize}
\item[$ \bullet $ ]	
 We define detection information (DI) as the mutual information between the received signal and the target state. DI is a theoretical metric for evaluating the performance of a target detection system.
\end{itemize}

\begin{itemize}
\item[$ \bullet $ ]	
We prove that the false alarm probability is equal to the prior probability of target existence as long as the observation interval is sufficiently large. This theorem establishes a connection between the NP criterion and the maximum a posteriori (MAP) criterion, laying the foundation for the application of Bayesian principles in the field of target detection.
\end{itemize}

\begin{itemize}
\item[$ \bullet $ ]
We propose the sampling \textit{a posteriori} (SAP) probability detector, which is a stochastic detector that exhibits detection performance approaching the theoretical DI.
\end{itemize}

\begin{itemize}
\item[$ \bullet $ ]
We prove that DI is the theoretical limit of target detection. Specifically, the empirical DI of SAP approaches the theoretical DI as the snapshots increase. Conversely, there is no detector whose empirical DI is greater than DI and the failure probability approaches zero.
\end{itemize}

\begin{itemize}
	\item[$ \bullet $ ]
We also prove the extended Fano's inequality, which plays a fundamental role in proving the converse theorem for the theoretical limit of target detection.
\end{itemize}

Numerical simulations are conducted to demonstrate that the empirical DI of the SAP detection can approach the DI with thousands of snapshots. The results also illustrate that the MAP and NP detection do not perform better than the theoretical limit.

The rest of the paper is organized as follows. 
In Section 2, the radar system model is introduced. The equivalent detection channel and the \textit{a posteriori} PDF are derived in Section 3 and Section 4, respectively. The definition of DI is shown in Section 5 and the theoretical false alarm probability and detection probability are provided in Section 6. Section 7 presents the false alarm theorem. Section 8 proposes the SAP detector and introduces the NP and the MAP detectors, as well as provides their performance evaluation. Section 9 proves the theoretical limit of the target detection. Section 10 gives the simulation results, and Section 11 concludes this paper.

\textit{Notations:}  Throughout the paper, $ a $  represents random variables and $ \boldsymbol{a} $ denotes a vector. ${\rm{E}}(\cdot)$ denotes the expectation. The operators $(\cdot)^{\mathrm{T}}$ and $(\cdot)^{\mathrm{H}}$ denote the transpose and the conjugate transpose of a vector, respectively. A Gaussian variable with expectation $\mu $ and variance $\sigma^2$ is denoted by $ {\cal C}\mathcal{N}(\mu ,\sigma^2)$. $I_{0}(\cdot)$ stands for the zero-order modified Bessel function of the first kind. $\pi(\cdot)$ denotes the \textit{a priori} probability of a random variable. $d(\cdot)$ denotes the detection function.  $ {\mathcal{R}}\left( \cdot \right)  $ denotes the real part of a complex scalar.

\section{System Model}
Consider a general radar system model with the base-band signal $ \psi(t) $ and a limited bandwidth $ B/2 $. Suppose there are $ K  $ random targets in the observation interval. The received signal is given by
\begin{equation}
	y\left( t \right) = \sum\limits_{k = 1}^K {{v_k}{s_k}\psi \left( {t - {\tau _k}} \right) + w\left( t \right)} ,
\end{equation}
where $ v_{k} $ is a discrete random variable characterizing the state of the $ k$-th target, that is, when $ v=1 $ or $ v=0 $, it corresponds to the target is existent or absent, respectively; $ s_{k} $ is the complex reflection coefficient; $ \tau_{k} $ is the propagation delay of the $ k$-th target; $ w(t) $ is complex Gaussian noise of a limited bandwidth $ B/2 $ with zero mean and $ N_{0}/2 $ variance in its real and imaginary parts, respectively. As shown in Figure 1, the radar detection system is equivalent to a joint modulation system according to (1). In the equivalent communication system, the target state $v$, the scattered signal $s$, and the time delay $\tau$ perform OOK modulation, amplitude-phase modulation, and time delay modulation on the baseband signal, respectively. The modulation processes are implemented during the process of reflecting radar signals from the target.

\begin{figure*}[!ht]
	\center{\includegraphics[width=13cm]  {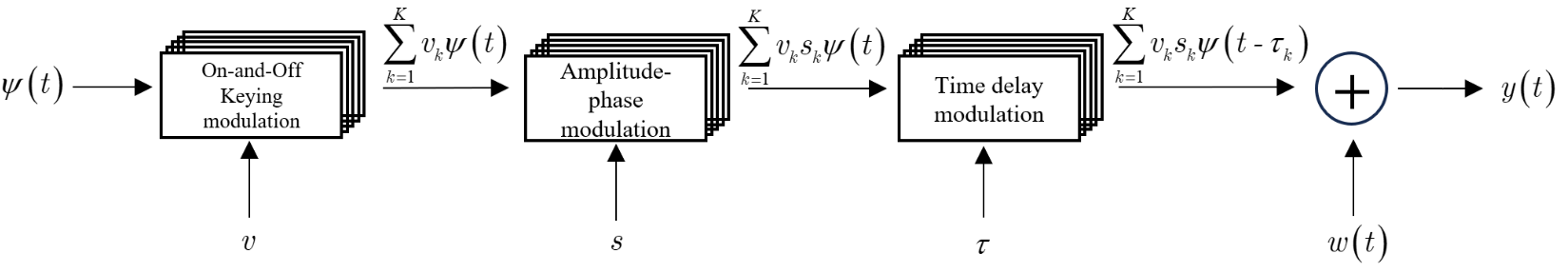}} 
	\caption{\label{1} Equivalent communication system model for radar detection systems}
\end{figure*}

Assume that the signal energy is almost completely within the observation interval. Sample the received signal $ y\left( t \right) $ at rate $ B $ according to the Nyquist-Shannon sampling theorem. The  corresponding sampling time is $ T $ in the regularized observation interval $ [-N/2, N/2) $, where $ N\!=\!T\!B $ is the total number of sampling points named the time-bandwidth product.  The discrete-time signal  can be expressed as 
\begin{equation}
	{{y}}\left( n \right) = \sum\limits_{k = 1}^K {{v_k}{s_k}\psi \left( {n - {x_k}} \right) + w\left( n \right)} ,n =  - \frac{T\!B}{2}, \cdots ,\frac{T\!B}{2} - 1,
\end{equation}
where $ x_{k}=B\tau_{k} $ is the regularized delay. The vector form of (2) is given by
\begin{equation}
	{\boldsymbol{y}} = {\boldsymbol{\Psi}}\left( x \right)diag\left( \boldsymbol{v} \right){\boldsymbol{s}} + {\boldsymbol{w}},
\end{equation}
where  ${\boldsymbol{\Psi}}\left( x \right) = \left[ {\cdots,{\boldsymbol{\psi}}\left( {{x_k}} \right), \cdots} \right] $ is the position matrix of the $ K $ targets; ${{\boldsymbol{\psi}}^{\rm{T}}}\left( {{x_k}} \right) = \left[ { \cdots ,\psi (n - {x_k}), \cdots } \right] $ is the sampling waveform of the $ k $-th target; ${\boldsymbol{v}} = ({v_1},{v_2}, \cdots ,{v_K}) $ is the target state vector;  ${\boldsymbol{s}} = {\left[ {{{s}_{1}},{{s}_{2}}, \cdots ,{{s}_{K}}} \right]^{\rm{T}}}$ is the scattering vector of the $ K$ targets; $ \boldsymbol{w} $ is a complex Gaussian noise vector, whose elements are independent identically distributed complex Gaussian variables with mean zero and variance $ N_{0} $.

\section{Equivalent Detection Channel}
According to the signal model in (3), the probability density function of $ {\boldsymbol{y}} $ conditioned on $ \boldsymbol{v} $, $ \boldsymbol{x} $, and $ \boldsymbol{s} $ is given by
\begin{equation}
	p\left( {{\boldsymbol{y}}\left| {{\boldsymbol{v}},{\boldsymbol{x,s}}} \right.} \right) \!=\! {\left( {\frac{1}{{{\rm{\pi }}{N_0}}}} \right)^N}\!\!\!\exp\! \left(\! { - \frac{1}{{{N_0}}}{{\left\| {{\boldsymbol{y}} - {\boldsymbol{\Psi}}{\rm{(}}{\boldsymbol{x}}{\rm{)}}diag\left( {\boldsymbol{v}}\right) {\boldsymbol{s}}} \right\|}^2}}\! \right),
\end{equation}
where  ${\boldsymbol{x}} \!=\! \left[ {x_1},{x_2}, \cdots ,{x_K}\right]$ denotes the regularized time delay vector.

To facilitate theoretical analysis, assume that there is only one target in the observation interval. The complex scattering signal is $s = \alpha {{\rm{e}}^{{\rm{j}}\varphi }} $, where $\alpha $ and $\varphi $ denote the amplitude and the phase, respectively. Assume that the sampling waveform is unit-energy signal, i.e.,$  {\left\| \boldsymbol{\psi}\right\|}^2=1 $. It is clear that $ v^{2}=v $, and expanding the exponent in (4) yields
\begin{equation}
	\begin{split}
		p\left( {\left. {\boldsymbol{y}} \right|v,x,\varphi } \right)
		={\left( {\frac{1}{{\pi {N_0}}}} \right)^N}\exp \left[ { - \frac{1}{{{N_0}}}\left( {{{\boldsymbol{y}}^{\rm{H}}}{\boldsymbol{y}} + v{\alpha ^2}} \right)} \right]
	\exp \left\{ {\frac{2}{{{N_0}}}{\mathcal{R}}\left[ {v\alpha {e^{ - j\varphi }}{{\boldsymbol{\psi}}^{\rm{H}}}\left( x \right){\boldsymbol{y}}} \right]} \right\}.
	\end{split}
\end{equation}
Futher, suppose that $\alpha $ is constant and $\varphi $ is distributed uniformly in $ [0,2\pi] $ (a Swerling 0 target-type model). Assuming that the delay is uniformly distributed over the observation interval $ [-T\!B/2, T\!B/2) $,   the equivalent detection channel is given by
\begin{equation}
	\begin{split}
		p\left( {{\bf{y}}\left| v \right.} \right) 
		={\rm{ }}{\left( {\frac{1}{{\pi {N_0}}}} \right)^N} \exp \left[ { - \frac{1}{{{N_0}}}\left( {{{\bf{y}}^{\rm{H}}}{\bf{y}}+v{\alpha ^2}} \right)} \right]
		\frac{1}{{T\!B}}\int_{ - \frac{{T\!B}}{2}}^{\frac{{T\!B}}{2}}  {{I_0}\left[ {\frac{2v\alpha }{{{N_0}}}\left| {{{\boldsymbol{\psi}}^{\rm{H}}}(x){\bf{y}}} \right|} \right] {\rm{d}}x},
	\end{split}
\end{equation}  where
\begin{equation}
	{I_0}\!\!\left[ {\displaystyle\frac{{2v\alpha }}{{{N_0}}}\left| {{{\boldsymbol{\psi}}^{\rm{H}}}\!\left( x \right)\!{\boldsymbol{y}}} \right|} \right] \! \!=\! \! \displaystyle\frac{1}{{2\pi }}\! \!\int_0^{2\pi } \!{\!\! \exp\! \left\{ {\!\!\frac{2}{{{N_0}}}\!{\mathcal{R}}\left[ {v\alpha {e^{ - j\varphi }}{{\boldsymbol{\psi}}^{\rm{H}}}\left( x \right){\boldsymbol{y}}} \right]}\! \right\}\! {\rm{d}}\varphi }.
\end{equation}  Note that $ \left| {{{\boldsymbol{\psi}}^{\rm{H}}}{\rm{(}}x{\rm{)}}{\boldsymbol{y}}} \right| $ is the module of output of the matched filter.

\section{\textit{A posteriori} probability distribution of the target state}
Based on the Bayesian formula and equivalent detection channel in (6), the \textit{a posteriori} PDF is given by
\begin{equation}
	\begin{split}
		P\left( {v\left| {\boldsymbol{y}} \right.} \right) = \frac{1}{\mathcal{Z} }\pi \left( v \right)\exp \left( { - v{\rho ^2}} \right)
		\frac{1}{{T\!B}}\int_{ - \frac{{T\!B}}{2}}^{\frac{{T\!B}}{2}} {{I_0}\left[ {2v{\rho ^2}{{\left| {{{\boldsymbol{\psi}}^{\rm{H}}}(x){\boldsymbol{y}}} \right|} \mathord{\left/
						{\vphantom {{\left| {{{\boldsymbol{u}}^{\rm{H}}}(x){\boldsymbol{y}}} \right|} \alpha }} \right.
						\kern-\nulldelimiterspace} \alpha }} \right]{\rm{d}}x}  ,
	\end{split}
\end{equation}where $ 	\mathcal{Z}  \! =\!\sum\limits_v \! \pi \!\left( v \right)\!\exp\! \left( { - v{\rho ^2}} \right)\displaystyle\!\! \frac{1}{{T\!B}}\!\!\int_{ - \frac{{T\!B}}{2}}^{\frac{{T\!B}}{2}} \!\!{{I_0}\!\!\left[ {2v{\rho ^2}{{\left| {{{\boldsymbol{\psi}}^{\rm{H}}}(x){\boldsymbol{y}}} \right|} \mathord{\left/
				{\vphantom {{\left| {{{\boldsymbol{u}}^{\rm{H}}}(x){\boldsymbol{y}}} \right|} \alpha }} \right.
				\kern-\nulldelimiterspace}\! \alpha }}\! \right]\!{\rm{d}}x}  $ denotes the regularized factor and $ {\rho ^2}\!=\!{{{\alpha ^2}} \mathord{\left/{\vphantom {{{\alpha ^2}} {{N_0}}}} \right.\kern-\nulldelimiterspace} {{N_0}}}  $  denotes the $ \mathrm{SNR} $. 

For a  snapshot, assuming the actual target state, position, and phase are $ v_{0} $, $ x_{0} $, and  $ \varphi _{0} $, respectively.  Then the received signal is expressed as
\begin{equation}
	{\boldsymbol{y}}_{0}\left( n \right) = v_{0}\alpha {{\rm e}^{j{\varphi _0}}}\psi\left({n-{x_0}}  \right) + {\boldsymbol{w}}_0\left( n \right).
\end{equation}
Substituting (9) into (8), the \textit{a posteriori} PDF is given by 
\begin{equation}	    
	P\left( {v\left| {\boldsymbol{y}} \right.} \right)\! \!=\!\! \frac{1}{\mathcal{Z}  }\pi \left( v \right)\exp\! \left( \!{ - v{\rho ^2}}\! \right)
	\!\frac{1}{{T\!B}}\!\int_{ - \frac{{T\!B}}{2}}^{\frac{{T\!B}}{2}} \!\!{{I_0}\!\!\left[ {2v{\rho} \!\!\left| {\sum\limits_{n =  - {{T\!B} \mathord{\left/
							{\vphantom {{T\!B} 2}} \right.
							\kern-\nulldelimiterspace} 2}}^{{{T\!B} \mathord{\left/
							{\vphantom {{T\!B} 2}} \right.
							\kern-\nulldelimiterspace} 2} - 1} \!\! \!{\left( \!\!{v_{0}\rho\psi \left( {n - {x_0}} \right)\psi (n - x)  \!\!+ \!\! \frac{1}{\alpha }\rho{{\rm{e}}^{ - j{\varphi _0}}}{{\boldsymbol{w}}_{{0}}}\left( n \right)\psi (n - x)} \! \right)} }\! \right|} \right]\!\!\!{\rm{d}}x}.
\end{equation}	 To simplify the exhibition, we assume that the transmitted signal is $ {\boldsymbol{\psi }}\left( x \right) = {\mathop{\rm sinc}\nolimits} \left( x \right) $, then the  \textit{a posteriori}  PDF is given by
\begin{equation}
	\begin{split}
		P\left( {v\left| {\boldsymbol{y}} \right.} \right) 
		=\frac{1}{\mathcal{Z}  }\pi \left( v \right)\exp \left( { - v{\rho ^2}} \right)
		\frac{1}{T\!B}\int_{ - \frac{{T\!B}}{2}}^{\frac{{T\!B}}{2}} {{I_0}\left( {2v{\rho }\left|  {v_{0}\rho{\mathop{\rm sinc}\nolimits}\left( {x-{x_0}} \right)+ \mu \left( x \right)}  \right|} \right){\rm d}x},		
	\end{split}		
\end{equation} where $ {\mathop{\rm sinc}\nolimits} \left( {x - {x_0}} \right)  = {\sum\limits_{n =  - {T\!B \mathord{\left/
				{\vphantom {N 2}} \right.
				\kern-\nulldelimiterspace} 2}}^{  {T\!B \mathord{\left/
				{\vphantom {N 2}} \right.
				\kern-\nulldelimiterspace} 2-1}}\psi\left( {n-{x_0}} \right)\psi(n-x)}$ is the auto-correlation function of base-band signals and $\mu\left( x \right) = \frac{1}{\alpha }\rho{{\rm e}^{ - j{\varphi _0}}}{{{w}}_0}\left( x \right)$ obeys the complex standard normal distribution, where	$ {{w}_0}\left( x \right) \!=\! \sum\limits_{n =  - {T\!B \mathord{\left/
			{\vphantom {N 2}} \right.
			\kern-\nulldelimiterspace} 2}}^{{T\!B \mathord{\left/
			{\vphantom {N 2}} \right.
			\kern-\nulldelimiterspace} 2}-1} {{\mathop{\rm sinc}\nolimits} \left( {n - x} \right){{\boldsymbol{w}}_0}(n)}  $ is complex white Gaussian noise with zero mean and  $ {{{N_0}} \mathord{\left/
		{\vphantom {{{N_0}} 2}} \right.
		\kern-\nulldelimiterspace} 2} $ variance in its real and imaginary parts, respectively. Equation (11) is also expressed as
\begin{equation}
	P\left( {v\left| {\boldsymbol{y}} \right.} \right) =\displaystyle \frac{{\pi (v){\Upsilon }\left( {v,{\boldsymbol{y}}} \right)}}{{\pi (0) + \pi (1){\Upsilon }\left( {1,{\boldsymbol{y}}} \right)}},
\end{equation} where 
\begin{equation}
	\begin{aligned}
		\Upsilon \left( {v,{\boldsymbol{y}}} \right) = \exp \!\left( { - v{\rho ^2}} \right)
		\displaystyle\frac{1}{T\!B}\int_{ - \frac{{T\!B}}{2}}^{\frac{{T\!B}}{2}}{{I_0}\left( {2v{\rho }\left| \!{v_{0}\rho{\mathop{\rm sinc}\nolimits}\left( {x - {x_0}} \right) +\mu\left( x \right) } \right|} \right)\!{\rm{d}}x} 
	\end{aligned}
\end{equation} denotes the detection statistics, which  is the average over the Bessel function with respect to time in the observation interval.

\section{Detection Information}
Information theory provides a means to quantify the reduction in \textit{a priori} uncertainty of a transmitted message.  The target detection is essentially the process of acquiring information about the target state, therefore the detection performance can be measured by mutual information in tne unit of $\mathrm{bit}$.
According to the equivalent detection channel in (6), the definition of DI is given as follows.

\definition[Detection Information] { The mutual information $I\left( {{\boldsymbol{y}};{v}} \right) $ between the target state $ v $ and the received signal $ \boldsymbol{y} $ is defined as the DI,
	\begin{equation}
		I\left( {{\boldsymbol{y}};{{v}}} \right) = {\rm{E}}_{ {\boldsymbol{y}}}\left[ {\log \frac{{p\left( {\left. {\boldsymbol{y}} \right|v} \right)}}{{p\left( {\boldsymbol{y}} \right)}}} \right],
	\end{equation}
	where $ 	p\left( {\boldsymbol{y}} \right) = \sum\limits_v^{} {p\left( {\left. {\boldsymbol{y}} \right|v} \right)\pi \left( v \right)}  $	is the probability density function of $ {\boldsymbol{y}} $.
}

The DI is a theoretical measure for the assessment of target detection. As a positive indicator, the more DI there is, the better the detection performance will be. The alternative calculation of the DI is given according to the \textit{a posteriori} PDF in (11), i.e.,
\begin{equation}
	I\left( {{\boldsymbol{y}};{{v}}} \right) = H\left( {{v}} \right) - H\left( {{{v}}|{\boldsymbol{y}}} \right),
\end{equation}
where 
\begin{equation}
	H\left( {{v}} \right) =  - \sum\limits_v^{} {\pi (v)\log \pi (v)}
\end{equation}
is the \textit{a priori}  entropy and
\begin{equation}
	H\left( {{{v}}|{\boldsymbol{y}}} \right) = {{\rm{E}}_{\boldsymbol{y}}}\left[ { - \sum\limits_v^{} {P(v|{\boldsymbol{y}})\log P(v|{\boldsymbol{y}})} } \right]
\end{equation}
is the \textit{a posteriori} entropy.

\section{Theoretical false alarm probability and detection probability}
The false alarm probability and detection probability are two common decision probabilities. We provide the theoretical expressions for both based on the \textit{a posteriori}  PDF of the target state below. We decompose the signal $ \boldsymbol{y} $ into two sub-signals ${{\boldsymbol{y}}_0}$ and ${{\boldsymbol{y}}_1 }$, which denote the received signal with only noise ($ v_{0}\!=\!0 $) and with both the target signal and the noise ($ v_{0}\!=\!1 $), respectively.  Then the outputs of the matched filter are
\begin{equation}
	\left\{ \begin{array}{l}
		{{\boldsymbol{\psi}}^{\rm{H}}}{\rm{(}}x{\rm{)}}{{\boldsymbol{y}}_0} = w{\rm{(}}x{\rm{),}}\\
		{{\boldsymbol{\psi}}^{\rm{H}}}{\rm{(}}x{\rm{)}}{{\boldsymbol{y}}_1} = {\rm{sinc(}}x - {x_0}{\rm{)}} + w{\rm{(}}x{\rm{)}}{\rm{.}}
	\end{array} \right.
\end{equation}
According to (13), the \textit{a posteriori} probabilities of existing state are
\begin{equation}
	\left\{ \begin{array}{l}
		P\left( {1{{\left| {\boldsymbol{y}}_0 \right.}}} \right) =\displaystyle \frac{{\pi (1)\Upsilon \left( {1,{{\boldsymbol{y}}_0}} \right)}}{{\pi (0) + \pi (1)\Upsilon \left( {1,{{\boldsymbol{y}}_0}} \right)}},\\
		P\left( {1{{\left| {\boldsymbol{y}}_1 \right.}}} \right) =\displaystyle \frac{{\pi (1)\Upsilon \left( {1,{{\boldsymbol{y}}_1}} \right)}}{{\pi (0) + \pi (1)\Upsilon \left( {1,{{\boldsymbol{y}}_1}} \right)}},
	\end{array} \right.
\end{equation}
where the corresponding detection statistics are given by
	\begin{equation}
	\left\{ {\begin{array}{*{20}{l}}
			\ {\Upsilon \left( {1,{{\boldsymbol{y}}_0}} \right) = \exp \left( { - {\rho ^2}} \right)\displaystyle\frac{1}{{T\!B}}\int_{ - \frac{{T\!B}}{2}}^{\frac{{T\!B}}{2}} {\!{I_0}\left( {2{\rho }\!\left| \mu\left( x \right) \right|} \right){\rm{d}}x} ,}\\
			\begin{array}{c}
				\begin{split}
					\Upsilon \left( {1,{{\boldsymbol{y}}_1}} \right) = \exp \left( { - {\rho ^2}} \right)
					\displaystyle\frac{1}{{T\!B}}\!\int_{ - \frac{{T\!B}}{2}}^{\frac{{T\!B}}{2}} \!{{I_0}\left( {2{\rho}\left| {\rho{\rm{sinc}}(x - {x_0})\! + \mu\left( x \right)} \right|} \right){\rm{d}}x}.
				\end{split}
			\end{array}
	\end{array}} \right.
\end{equation} The  theoretical false alarm probability and detection probability are
\begin{equation}
	\left\{ \begin{array}{c}
		{P_{{\rm{FA}}}} = {{\rm{E}} _{{\boldsymbol{y}}_0}}\left[ {\displaystyle\frac{{\pi (1)\Upsilon \left( {1,{{\boldsymbol{y}}_0}} \right)}}{{\pi (0) + \pi (1)\Upsilon \left( {1,{{\boldsymbol{y}}_0}} \right)}}} \right],\\
		{P_{\rm{D}}} = {{\rm{E}} _{{\boldsymbol{y}}_1}}\left[ {\displaystyle\frac{{\pi (1)\Upsilon \left( {1,{{\boldsymbol{y}}_1}} \right)}}{{\pi (0) + \pi (1)\Upsilon \left( {1,{{\boldsymbol{y}}_1}} \right)}}} \right].
	\end{array} \right.
\end{equation}

As we can see from (21) that $ P_{{\rm{FA}}} $ and $ P_{{\rm{D}}} $ depend on the \textit{a priori} probability $ \pi(1) $ of the existing state. In turn, we can obtained the $ \pi(1) $ from setting the tolerable level of  $ P_{{\rm{FA}}} $ according to the specific needs in detection environment. In particular,  it is shown that, when the SNR is 0, $ P_{{\rm{FA}}} $ and $ P_{{\rm{D}}} $ are equal to $ \pi(1) $, respectively. This is attributed to both $ \Upsilon \left( {1,{{\boldsymbol{y}}_0}} \right) $ and $ \Upsilon \left( {1,{{\boldsymbol{y}}_1}} \right) $ are 1 as shown in (20). We can refer to this phenomenon as unbiased detection.                                                                                                             

\subsection{Special Case}

We consider a special case in which the matched filter is perfectly matched to the target position. In this case, the correlation peak is the highest when the SNR is high, so it is reasonable to use the peak value for target detection. According to (6), the probability density function of $ \boldsymbol{y} $ conditioned on the target state $ v $ and the position $ x_{0} $ is given by
\begin{equation}
	\begin{split}
		p\left( {\left. {\boldsymbol{y}} \right|v,x_{0}} \right) = {\left( {\frac{1}{{\pi {N_0}}}} \right)^N}\exp \left[ { - \frac{1}{{{N_0}}}\left( {{{\boldsymbol{y}}^{\rm{H}}}{\boldsymbol{y}} + v{\alpha ^2}} \right)} \right]
	 {I_0}\left[ {\frac{{2v\alpha }}{{{N_0}}}\left| {{{\boldsymbol{\psi}}^{\rm{H}}}\left( x_{0} \right){\boldsymbol{y}}} \right|} \right],
	\end{split}
\end{equation} By the Bayesian formula and (22), the  \textit{a posteriori} PDF is given by 
\begin{equation}
	\begin{split}
		P\left( {v\left| {{\boldsymbol{y}},{x_0}} \right.} \right)\! = \!\frac{1}{\mathcal{Z} }\!\pi \left( v \right)\exp \left( { - v{\rho ^2}} \right)\!{I_0}\!\!\left[ {\frac{{2v\alpha }}{{{N_0}}}\left| {{{\boldsymbol{\psi}}^{\rm{H}}}\left( {{x_0}} \right){\boldsymbol{y}}} \!\right|} \!\right]
	\end{split}
\end{equation}
or
\begin{equation}
	P\left( {v\left| {\boldsymbol{y}},{x_0} \right.} \right) =\displaystyle \frac{{\pi (v){\Upsilon }\left( {v,{\boldsymbol{y}},{x_0}} \right)}}{{\pi (0) + \pi (1){\Upsilon}\left( {1,{\boldsymbol{y}}},{x_0} \right)}},
\end{equation}
where the detection statistics  are
\begin{equation}
	\left\{ \begin{array}{l}
		\Upsilon \left( {1,{{\boldsymbol{y}}_0},{x_0}} \right) \!=\! \exp \left( { - {\rho ^2}} \right)\!{I_0}\left( {2{\rho }\left|\mu\left( x \right) \right|} \right),\\
		\Upsilon \left( {1,{{\boldsymbol{y}}_1},{x_0}} \right) \!=\! \exp \left( { - {\rho ^2}} \right)\!{I_0}\left( {2{\rho }\left|  \!{\rho \!+ \mu\left( x \right)} \right|} \right).
	\end{array} \right.
\end{equation}
As a result, the DI is 
\begin{equation}
	I\left( {\boldsymbol{y}};{v}|x_{0} \right) \!=\! H\left( {{v}} \right) \!-\! {{\rm{E}} _{\boldsymbol{y}}}\!\!\left[ { - \sum\limits_v^{}\!\! {P(v|{\boldsymbol{y}},x_{0})\log P(v|{\boldsymbol{y}},x_{0})} } \right],
\end{equation} The  theoretical false alarm probability and detection probability are
\begin{equation}
	\left\{ \begin{array}{c}
		{P_{{\rm{FA}}}} = {{\rm{E}} _{{\boldsymbol{y}}_0}}\left[ {\displaystyle\frac{{\pi (1)\Upsilon \left( {1,{{\boldsymbol{y}}_0}},{x_0} \right)}}{{\pi (0) + \pi (1)\Upsilon \left( {1,{{\boldsymbol{y}}_0}},{x_0} \right)}}} \right],\\
		{P_{\rm{D}}} = {{\rm{E}} _{{\boldsymbol{y}}_1}}\left[ {\displaystyle\frac{{\pi (1)\Upsilon \left( {1,{{\boldsymbol{y}}_1}},{x_0} \right)}}{{\pi (0) + \pi (1)\Upsilon \left( {1,{{\boldsymbol{y}}_1}},{x_0} \right)}}} \right].
	\end{array} \right.
\end{equation}

\section{False Alarm Theorem}
Both the NP and DI criteria are applicable for performance evaluation in target detection. The NP criterion relies on the given $ {P_{{\rm{FA}}}}  $, while the DI criterion depends on the $ \pi(1) $.  The following theorem points out the relationship between the $ {P_{{\rm{FA}}}}  $ and the $ \pi(1) $, allowing the NP and DI criteria to be allowed to be compared.
\theorem[False Alarm Theorem]{
	Let the target position be uniformly distributed on the observed interval. If the observation interval is sufficiently large, the false alarm probability is equal to the \textit{a priori} probability of target existence for arbitrary scattering characteristics in a complex additive white Gaussian noise channel, 
	\begin{equation}
		{P_{{\rm{FA}}}} = \pi (1).
	\end{equation}
}
\noindent \textit{Proof: See Appendix A.}

In a single snapshot, tens of thousands of points or even more can be sampled by modern radar detection systems. Therefore, it ensures that the condition of a sufficiently large observation interval in the false alarm theorem can be satisfied.

\section{Detectors and Performance Evaluation}
In this section, we propose the SAP detector. Prior to introducing SAP detector, we analyze the detection performance of MAP and NP detectors in the framework of information theory.

\subsection{Neyman-Pearson detector}
The NP detection maximizes the $ {P_{{\rm{D}}}}  $ for a given $ {P_{{\rm{FA}}}}  $ constraint $ \eta \left( {\eta  \in \left( {0,1} \right)} \right) $ by the likelihood ratio test. The likelihood ratio is
\begin{equation}
	\frac{{p({\boldsymbol{y}}|1)}}{{p({\boldsymbol{y}}|0)}}\begin{array}{*{20}{c}}
		{{H_1}}\\
		\mathbin{\lower.3ex\hbox{$\buildrel>\over
				{\smash{\scriptstyle<}\vphantom{_x}}$}} \\
		{{H_{\rm{0}}}}
	\end{array}{T_h},
\end{equation}
where $ T_{h} $ denotes the detection threshold given by the $ {P_{{\rm{FA}}}}  $. 

The detection threshold for a specific position $ x_{0} $ is given as \cite{12}
\begin{equation}
	{T_h} =  - \sqrt {{N_0}\ln {P_{{\rm{FA}}}}} ,
\end{equation}
and the detection probability is 
\begin{equation}
	{P_{\rm{D}}} = {Q_{\rm{M}}}\left( {\sqrt {2{\rho ^2}} ,\sqrt { - 2\ln {P_{{\rm{FA}}}}} } \right),
\end{equation}
where $ {Q_{\rm{M}}}\left(  \cdot  \right)$ is the Marcum function.

According to (6), $ \displaystyle{{p({\boldsymbol{y}}|1)} \mathord{\left/{\vphantom {{p({\boldsymbol{y}}|1)} {p({\boldsymbol{y}}|0)}}} \right. \kern-\nulldelimiterspace} {p({\boldsymbol{y}}|0)}}= \Upsilon\left( 1, \boldsymbol{y} \right) $  is easily obtained, where $ \Upsilon\left( 1, \boldsymbol{y} \right) $ denotes the detection statistics as shown in (13). The false alarm probability is $ {P_{{\rm{FA}}}} = P\left\{ {{\Upsilon}\left( {1,{{\boldsymbol{y}}_0}} \right) > {T_h}} \right\} $. It is worth noting that the detection statistics $ {{\Upsilon}\left( {1,{{\boldsymbol{y}}_0}} \right)}$ equals $ 1 $ as shown in (A9), and not a random variable. This means that the detection threshold cannot be gained,  resulting in a failure of the NP detection. Therefore, the performance comparison of detectors is only provided in the matching case presented in section 4.4.

M. Kondo\cite{34} provides the calculation method of the empirical information of the NP detection, which is based on the false alarm probability and detection probability, 
\begin{equation}
	\begin{split}
		I_{\rm{NP}} = H\left( {{v}} \right)  -\left[ {\mathcal{P}}H\left( A \right) -\left(  {1-\mathcal{P}}\right) H\left( D \right)\right] ,
	\end{split}
\end{equation}
where 
\begin{equation}
	\left\{ \begin{array}{l}
		{\cal P} = \pi \left( 1 \right){P_{\rm{D}}} +\pi \left( 0 \right) {P_{{\rm{FA}}}},\\
		A = {{\pi \left( 1 \right){P_{\rm{D}}}} \mathord{\left/
				{\vphantom {{\pi \left( 1 \right){P_{\rm{D}}}} {\cal P}}} \right.
				\kern-\nulldelimiterspace} {\cal P}},\\
		D = {{\pi \left( 0 \right)\left( {1 - {P_{{\rm{FA}}}}} \right)} \mathord{\left/
				{\vphantom {{\pi \left( 0 \right)\left( {1 - {P_{{\rm{FA}}}}} \right)} {\left( {1 - {\cal P}} \right)}}} \right.
				\kern-\nulldelimiterspace} {\left( {1 - {\cal P}} \right)}}.
	\end{array} \right.
\end{equation}

\subsection{Maximum \textit{a posteriori} probability detector}
The MAP detection outputs the decision that maximizes the \textit{a posteriori}  probability, 
\begin{equation}
	{\hat v_{{\rm{MAP}}}} = \arg \mathop {{\rm{max}}}\limits_v P(v|{\boldsymbol{y}}),
\end{equation}
where $P(v|{\boldsymbol{y}}) \!=\! {{p({\boldsymbol{y}}|v)\pi (v)} \mathord{\left/
		{\vphantom {{p({\boldsymbol{y}}|v)\pi (v)} {p({\boldsymbol{y}})}}} \right.
		\kern-\nulldelimiterspace} {p({\boldsymbol{y}})}} $. In the case of equal \textit{a priori} probabilities, the MAP detection can be seen as a regularization of the maximum likelihood (ML) detection, i.e., $ 	{\hat v_{{\rm{ML}}}} \!= \!\arg \mathop {{\rm{max}}}\limits_v p\left( {{\boldsymbol{y}}|v} \right) $. However, it is rare that every target has equal \textit{a priori} probability in a radar target detection. Hence, ML detection does not apply to the detection of radar targets.

The false alarm probability and the detection probability of the MAP detection are given by
\begin{equation}
	\left\{ \begin{array}{l}
		{P_{{\rm{FA}}}} = P\left\{ {p\left( {\left. 1 \right|{{\boldsymbol{y}}_0}} \right) > p\left( {\left. 0 \right|{{\boldsymbol{y}}_0}} \right)} \right\}\\
		{P_{\rm{D}}} = P\left\{ {p\left( {\left. 1 \right|{{\boldsymbol{y}}_1}} \right) > p\left( {\left. 0 \right|{{\boldsymbol{y}}_1}} \right)} \right\}
	\end{array} \right.
\end{equation}

The empirical information of the MAP detection is obtained by (32).

\subsection{Sampling \textit{a posteriori}  detector}

A general target detection system is denoted by $\left( {\mathbb{A}, \pi (v), p({\boldsymbol{y}}|v), {\mathbb{B}} , \hat v = d(\cdot)} \right)$, where $\mathbb{A} \left( \mathbb{A} \!=\! \left\{  0,1\right\} \right)  $ is an input set;  $ {\mathbb{B}}\left(  {\boldsymbol{y}} \!\in\! \mathbb{B}\right)  $ is a vector space in the complex regime; $\hat v = d(\cdot)$ is  a detector that is a function of the received signal $ \boldsymbol{y} $, and $ p({\boldsymbol{y}}|v) $ is the detection channel.  Here, all statistical properties of the target state and the detection channel are assumed to be known by the detector, although only a portion of those is known in practice. The \textit{a posteriori} PDF $  P\left( {v\left| {\boldsymbol{y}} \right.} \right) $ is used to detect targets since all these statistical properties are contained in the  $  P\left( {v\left| {\boldsymbol{y}} \right.} \right) $. Based on the $  P\left( {v\left| {\boldsymbol{y}} \right.} \right) $, we propose SAP detection. Specifically, the decision of the SAP detection satisfy

\begin{equation}
	{\hat v_{{\rm{SAP}}}} = \arg \mathop {\rm{smp}}\limits_v P\left( {v\left| {\boldsymbol{y}} \right.} \right),
\end{equation}	where $  \mathop {\rm{smp}}\limits_v \left\{  \cdot  \right\} $ denotes the sampling, which is the draw of target states based on the \textit{a posteriori} probabilities. In Table  \uppercase\expandafter{\romannumeral1}, the comparison of decisions made by MAP and SAP detection is given to illustrate the sampling operation.

\begin{table*}[htbp]
	\centering
	{\caption{\small {{Decisions of  MAP And SAP Detection for 10 Snapshots}}}
		
		\begin{tabular}{p{2.5cm}p{0.7cm}p{0.5cm}p{0.5cm}p{0.5cm}p{0.5cm}p{0.5cm}p{0.5cm}p{0.5cm}p{0.5cm}p{0.5cm}p{0.5cm}}
			\toprule[0.5mm]
			\specialrule{0em}{4pt}{4pt}
			
			\multirow{2}{*}[-2.2ex]{{$\begin{array}{l}
						P\left( {1\left| {\boldsymbol{y}} \right.} \right) = 0.7\\
						P\left( {0\left| {\boldsymbol{y}} \right.} \right) = 0.3
					\end{array}$} }
			& {{$\hat v_{{\rm{MAP}}}$} }& { 1} & { 1}  & { 1}  & {1}  & { 1}  & { 1}  & { 1}  & { 1}  & {1}  & { 1}   \\  \specialrule{0em}{4pt}{4pt}\cline{2-12} 
			\specialrule{0em}{4pt}{4pt}
			& {{$\hat v_{{\rm{SAP}}}$} } & { 1}  & { 1}  & {0}  & { 1}  & { 0}  & { 1}  & {1}  & { 0}  & { 1}  & { 1}  \\ 
			\specialrule{0em}{4pt}{4pt} 
			\bottomrule[0.5mm]
	\end{tabular}}
\end{table*}

As can be seen that the MAP detection always selects 1 according to the maximum value of the \textit{a posteriori} probability of the target state, i.e., $ P\left( {1\left| {\boldsymbol{y}} \right.} \right)  $.  On the contrary, the decisions of SAP detection have seven 1s and three 0s in ten snapshots, which means that the decisions are drawn based on $ P\left( {v\left| {\boldsymbol{y}} \right.} \right)  $. The SAP detection is stochastic and theoretical, as demonstrated by the variation in decisions given the received signal $ \boldsymbol{y} $.

\section{Target Detection Theorem}
Inspired by Shannon's coding theorem, we prove the  target detection theorem from the achievability and the converse results in this section. The major difference with the coding theorem is that the receiver does not possess pre-knowledge of all possibilities for the transmitted sequence. Therefore, the proof given below makes use of the properties of the typical set, whose essential idea is random detection. 

Consider a $ m $ extension of  $\left( {\mathbb{V},\pi (v),p({\boldsymbol{y}}|v),\mathbb{Y}} \right) $ is denoted as $\left( {\mathbb{V}^{m},\pi (v^{m}),p({\boldsymbol{y}}^{m}|v^{m}),\mathbb{Y}^{m}} \right) $, where  $ \mathbb{V}^{m} $ denotes the extended input set,  $ \pi (v^{m}) $ represents the  \textit{a priori} distribution of the extended target state $ v^{m} $ ( $ {v^m} = ({v_1},{v_2}, \ldots ,{v_m}) $ ), $ p({\boldsymbol{y}}^{m}|v^{m}) $ denotes the extended channel ( $ {\boldsymbol{y}}^m = ({{\boldsymbol{y}}_1},{{\boldsymbol{y}}_2}, \ldots ,{{\boldsymbol{y}}_m})$ ), and $ \mathbb{\boldsymbol{Y}}^{m} $ is the extended signal set. The target detection system for $ m $ snapshots is as shown in Fig. 2. As can be seen, $m$ extended received signals are gained from the extended target states through the extended detection channel. The SAP detectors perform detection according to the $m$ extended received signals. The empirical entropy is calculated based on the  $m$ decisions to evaluate the detection performance.

\begin{figure*}[!ht]
	\center{\includegraphics[width=10cm]  {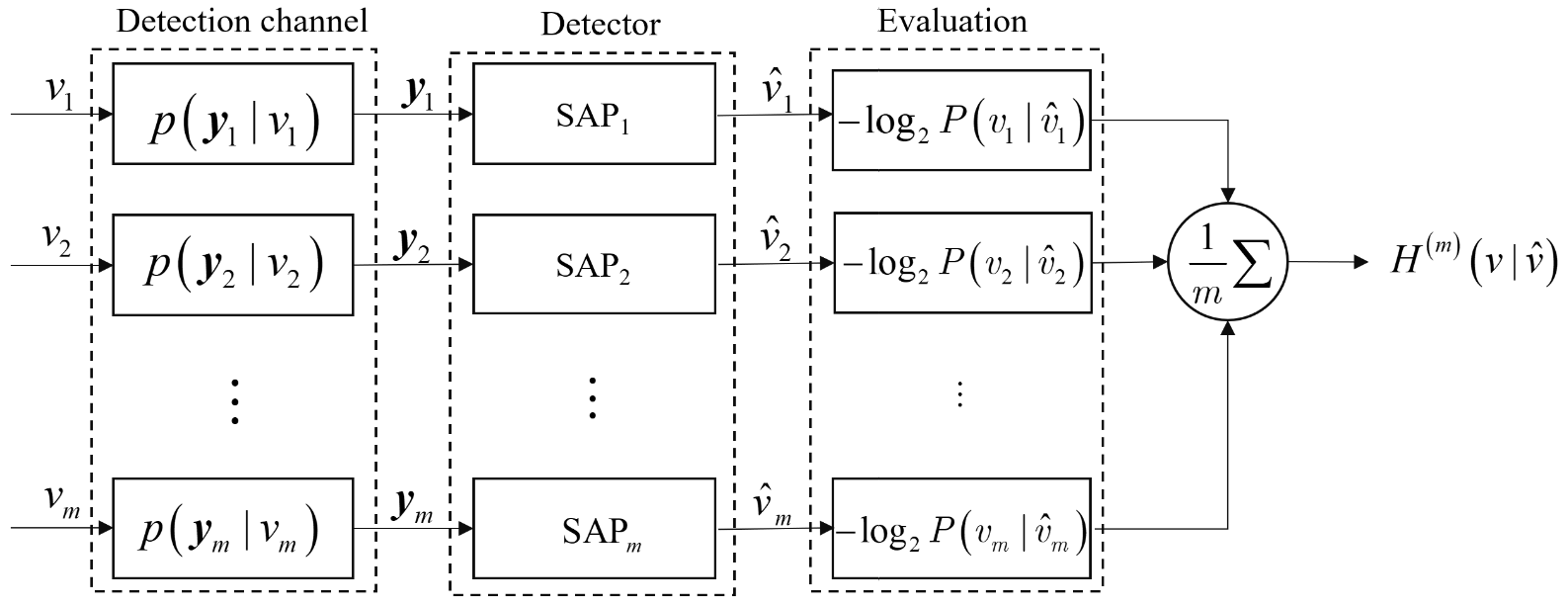}} 
    \caption{\label{1} Target detection system for $ m $ snapshots}
\end{figure*}

It can be seen that $\left( {{v^m},{\boldsymbol{y}}^{m},{{\hat v}^m}} \right) $  forms a Markov chain. The  \textit{a priori} distribution $ v^{m} $ and the extended channel $  p({{\boldsymbol{y}}^m}|{v^m}) $ satisfy 
\begin{equation}
	\pi \left( {{v^m}} \right) = \prod\limits_{i= 1}^m {\pi \left( {{v_{i}}} \right)} ,
\end{equation}
\begin{equation}
	p({{\boldsymbol{y}}^m}|{v^m}) = \prod\limits_{i = 1}^m {p({{\boldsymbol{y}}_{i}}|{v_{i}})} . 
\end{equation}

Let's take the example of evolutionary selection in a biological community to explain the target detection theorem. Assuming that each individual in the community can independently choose between two evolutionary paths, the theorem states that, as long as the number $m$ of individuals is sufficiently large, in the sense of empirical information approaching theoretical DI, the randomly selected evolutionary mode is optimal.

Before proving the theorem, we introduce the definitions and lemmas required for developing the target detection theorem.

\lemma [Chebyshev Law of Large Numbers] { If the random sequence ${z_1},{z_2}, \cdots {z_m} $ are independent identically distributed (i.i.d.) with mean $ \mu  $ and variance ${\sigma ^2} $, where the sample mean is ${\bar z_m} =\displaystyle \frac{1}{m}\sum\nolimits_{i = 1}^m {{z_i}} $, then
	\begin{equation}
		\Pr \left\{ {\left| {{{\bar z}_m} - \mu } \right| > \varepsilon } \right\} \le \frac{{{\sigma ^2}}}{{m{\varepsilon ^2}}}.
	\end{equation}	
}

\lemma [The Asymptotic Equipartition Property (AEP)] { If the random sequence ${v_1},{v_2}, \cdots {v_m} $ are i.i.d. $~ \pi\left(v \right)  $, then
	\begin{equation}
		\begin{split}
			- \frac{1}{m}\log \pi ({v_1},{v_2}, \cdots ,{v_m}) 	\to {\rm{E}} \left[ { - \log \pi (v)} \right] 
			= H(v) \quad
			in \ probability. 
		\end{split}
	\end{equation}
}

\definition[Typical Set] {If the random sequence ${v_1},{v_2}, \cdots {v_m} $ are i.i.d. ${\sim} \pi \left( v \right) $, for any $\varepsilon  > 0$, the typical set is defined as
	\begin{equation}
		\begin{array}{c}
			\begin{aligned}
				\mathbb{A}_\varepsilon ^{(m)}(v) = \left\{ {\left( {{v_1},{v_2}, \cdots ,{v_m}} \right) \in {v^m}} \right.:
				\left|\! { - \frac{1}{m}\log \pi \left( {{v_1},{v_2}, \cdots ,{v_m}} \right) \!-\! H(v)}\! \right|
				\!<\! \left. \varepsilon  \right\}
			\end{aligned}
		\end{array}
	\end{equation}where $\pi \left( {{v_1},{v_2}, \cdots ,{v_m}} \right) = \prod\limits_{i = 1}^m {\pi ({v_i})}  $.
} 
As a result of the AEP, the typical set  $ \mathbb{A} _\varepsilon ^{(m)}(v) $ has the following properties:
\lemma {
	For any $\varepsilon  > 0$, when $ m $ is sufficiently large, there is
	
	$ \left( 1\right)  $ 
	$ \Pr \left\{ {v \in \mathbb{A}_\varepsilon ^{(m)}(v)} \right\} > 1 - \varepsilon  $
	
	$ \left( 2\right)  $ 
	${2^{ - m\left[ {H(v) + \varepsilon } \right]}} < \pi (v) < {2^{ - m\left[ {H(v) - \varepsilon } \right]}} $
	
	$ \left( 3\right)  $ 
	$(1 - \varepsilon ){{\rm{2}}^{m(H(v) + \varepsilon )}}{\rm{ < }}\left\| {\mathbb{A} _\varepsilon ^{(m)}(v)} \right\| < {{\rm{2}}^{m(H(v) + \varepsilon )}}$, where $ \left\| {\mathbb{A} _\varepsilon ^{(m)}(v)} \right\|  $ denotes the cardinality of the set $ \mathbb{A} $, i.e., the number of elements in the set $ \mathbb{A} $.
}
\noindent\textit{Proof: see ([18], pp. 51-53).}

\definition [Jointly Typical Sequences]	{The set $ \mathbb{A}_\varepsilon ^{\left( m \right)} (v,\boldsymbol{y})$ of jointly typical sequences $\left( {{v^m},{{\boldsymbol{y}}^m}} \right)$ with respect to the distribution $p(v,{\boldsymbol{y}})$ is the set of $ M $-sequences with an empirical entropy $ \varepsilon $ close to the true entropy, i.e., 
	\begin{equation}
		\begin{split}
			\begin{array}{l}
				\mathbb{A}_\varepsilon ^{\left( m \right)}(v,\boldsymbol{y}) = \Bigg\{ \left( {{v^m},{{\boldsymbol{y}}^m}} \right) \in {\mathbb{A}^m}\! \times {\mathbb{B}^m}:\\
				\left| { - \displaystyle\frac{1}{m}\sum\limits_{i = 1}^m {\log P\left( {{v_{i}}} \right)}  - H\left( v \right)} \right| < \varepsilon \\
				\left| { - \displaystyle \frac{1}{m}\sum\limits_{i = 1}^m {\log p\left( {{{\boldsymbol{y}}_{i}}} \right)}  - H\left( y \right)} \right| < \varepsilon \\
				\left| { - \displaystyle \frac{1}{m}\sum\limits_{i = 1}^m {\log p\left( {{v_{i}},{{\boldsymbol{y}}_{i}}} \right)}  - H\left( {v,{\boldsymbol{y}}} \right)} \right| < \varepsilon \Bigg\} ,
			\end{array}
		\end{split}
	\end{equation}
} 
\noindent where
\begin{equation}
	p({{\boldsymbol{y}}^m},{v^m}) = \prod\limits_{i = 1}^m {p({{\boldsymbol{y}}_{i}},{v_{i}})} . 
\end{equation}

\lemma[Joint AEP] { Let the sequence $ ({v^m},{\boldsymbol{y}^m}) $ are i.i.d.  $ {\sim}  p({v^m},{\boldsymbol{y}}^m) = \prod\limits_{i = 1}^m {p({v_i},{\boldsymbol{y}}_i)}  $, for any $\varepsilon  > 0$, if $ m $  is sufficiently large, then
	
	$ \left( 1\right)  $ 
	$\Pr \left\{ {({v^m},{\boldsymbol{y}^m}) \in \mathbb{A}_\varepsilon ^{\left( m \right)}\left( {v;\boldsymbol{y}} \right)} \right\} \ge 1 - \varepsilon  $
	
	$ \left( 2\right)  $ 
	$ {2^{ - m\left[ {h(v,\boldsymbol{y}) + \varepsilon } \right]}} < p(v,{\boldsymbol{y}}) < {2^{ - m\left[ {h(v,\boldsymbol{y}) - \varepsilon } \right]}}  $
}
\noindent\textit{Proof: see ([18], pp. 195-197).}

\lemma[Conditional AEP]{   For any $\varepsilon  > 0$, if $ m $  is sufficiently large, then
	
	$ \left( 1\right)  $ 
	${2^{ - m\left( {H(v|\boldsymbol{y}) + 2\varepsilon } \right)}} < P(v|{\boldsymbol{y}}) < {2^{ - m\left( {H(v|\boldsymbol{y}) - 2\varepsilon } \right)}},  $
	
	\hspace{0.42cm} ${2^{ - m\left( {h(\boldsymbol{y}|v) + 2\varepsilon } \right)}} < p({\boldsymbol{y}}|v) < {2^{ - m\left( {h(\boldsymbol|v) - 2\varepsilon } \right)}},$ where $(v,{\boldsymbol{y}}) \in 	\mathbb{A}_\varepsilon ^{(m)}(v,\boldsymbol{y}) $.
	
	$ \left( 2\right)  $ Let $ 	\mathbb{A}_\varepsilon ^{(m)}(v|\boldsymbol{y}) = \left\{ {v:(v,{\boldsymbol{y}}) \in 	\mathbb{A}_\varepsilon ^{(m)}(v,\boldsymbol{y})} \right\} $ be a set of all $ v^{m} $ that form jointly typical sequences with $ \boldsymbol{y}^{m} $, then 
	\begin{equation}
		\left( {1\! -\! \varepsilon } \right){2^{m\left( {H(v|\boldsymbol{y}) - 2\varepsilon } \right)}} \!<\!\! \left\| {	\mathbb{A}_\varepsilon ^{(m)}(v|\boldsymbol{y})} \right\| \!\!<\! {2^{m\left( {H(v|\boldsymbol{y}) + 2\varepsilon } \right)}}.
	\end{equation}
	
	$ \left( 3\right)  $ Let $ v^{m} $ be the jointly typical sequences with $ \boldsymbol{y}^{m} $, then
	\begin{equation}
		\Pr \left( {{v^m} \in 	\mathbb{A}_\varepsilon ^{(m)}(v|\boldsymbol{y})} \right) > 1 - 2\varepsilon.
	\end{equation}
}
\noindent\textit{Proof: see([19], pp.212-214).}

The typical sets and the conditional typical set given the received sequence $ \boldsymbol{y}^{m} $ are as shown in Fig. 3. $ \left\| {{v^m}} \right\| $ and $ \left\| {{\boldsymbol{y}^m}} \right\| $ essentially denote the infinite sets contain all of the possible sequences of  $ v^{m} $  and $ \boldsymbol{y}^{m} $, respectively. Here, the two are considered finite sets. $\mathbb{A}_\varepsilon ^{(m)}\left( {v} \right) $ and $\mathbb{A}_\varepsilon ^{(m)}\left( {\boldsymbol{y}} \right) $ represent the typical sets contain all typical sequences of  $ v^{m} $  and $ \boldsymbol{y}^{m} $, respectively. The sequence  $ v^{m} $ is transmitted through the detection channel $ p\left( \boldsymbol{y}^{m}|v^{m}\right)  $ to the receiver, who then makes a decision $ {\hat{v} }^{m}$ based on the typical set detection. Specifically, the typical set detection finds  the $ {\hat{v} }^{m}$ that  forms the jointly typical sequence with $ \boldsymbol{y}^{m} $ given $ \boldsymbol{y}^{m} $. The cardinality of the typical set  $\mathbb{A}_\varepsilon ^{(m)}\left( { v|\boldsymbol{y}} \right) $ reflects the performance of target detection  since $ {\hat{v} }^{m}$ is not unique as illustrated in Fig. 3. Therefore, unlike in Shannon's coding theorem, we cannot evaluate detection results in terms of correct or not. Rather, successful and failed decisions are used as detection results, and the detection performance is measured only when the decision is successful. Below, we give the two concepts of successful and failed decisions.

\begin{figure} 
	\centering 
	\includegraphics[width=3.5in]{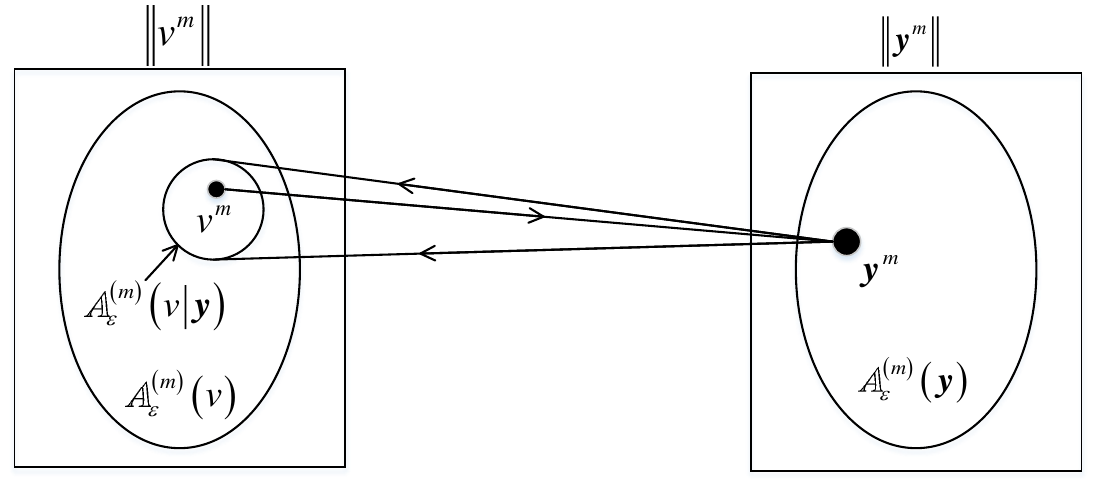} 
	\caption{The typical sets and the conditional typical set for the target detection system.} 
	\label{fig:side:b} 
\end{figure}

\definition[Detection Failure] {	If the decision sequence satisfies $ {\hat v^m} \notin	\mathbb{A}_\varepsilon ^{(m)}\left( { v|{{\boldsymbol{y}}}} \right) $,  the detection is said to be failed. The probability of detection failure is $ P_{f}^{\left( m\right) }=Pr\left\lbrace {\hat v^m} \notin	\mathbb{A}_\varepsilon ^{(m)}\left( { v|{{\boldsymbol{y}}}} \right)  \right\rbrace  $. Conversely, 	if $ {\hat v^m} \in 	\mathbb{A}_\varepsilon ^{(m)}\left( { v|{{\boldsymbol{y}}}} \right) $, the detection is said to be successful. 
} 

It is noted that uncertainty remains after a successful detection is made. Define a binary random variable $ E $, where $ E=0 $ and $ E=1 $ represent the event that successful and failed decisions are made, respectively. That is,

\begin{equation}
	E = \left\{ \begin{array}{l}
		0,{{\hat v}^m} \in  \mathbb{A}_\varepsilon ^{\left( m \right)}(v|{\boldsymbol{y}}),\\
		1,{{\hat v}^m} \notin \mathbb{A}_\varepsilon ^{\left( m \right)}(v|{\boldsymbol{y}}).
	\end{array} \right.
\end{equation} Hence, the failed probability is $P_f^{(m)} = \Pr \left( {E = 1} \right) $.

For a specific detector, the uncertainty regarding   given detection result  is evaluated by its empirical entropy $ - {\log _2}P\left( {{v_i}|{{\hat v}_i}} \right)$.
\definition [Empirical Entropy] If the detection is successful, i.e., $ E=0 $, we refer to 
	\begin{equation}
		{H^{(m)}}(v|\hat v,E = 0) = \frac{1}{m}\sum\limits_i { - {{\log }_2}P\left( {{v_i}|{{\hat v}_i}} \right)},
	\end{equation} is the empirical entropy.

In practice, it is generally not feasible to obtain the theoretical expression of the \textit{a posteriori} probability distribution $P\left( {{v_i}|{{\hat v}_i}} \right)$, making it impossible to directly compute the empirical entropy. However, the upper and lower bounds of the empirical entropy can be determined through the use of typical sets.

According to the maximal discrete entropy theorem, the empirical entropy satisfies

\begin{equation}
	{H^{(m)}}( v|\hat v, E=0) \le \frac{1}{m}\log  {\left\| {\mathbb{A}_\varepsilon ^{(m)}\left( {v|{{\boldsymbol{y}}}} \right)} \right\|} .
\end{equation}  The empirical entropy is regarded as a negative indicator of target detection, i.e., the larger empirical entropy, the less accurate performance of detection.

\definition[Empirical DI] {	 If the decision is successful, i.e., $ E=0$, the empirical DI of target detection is defined as 
	\begin{equation}
		{I^{(m)}}\left( v;\hat v| E=0 \right) = H(v) - {H^{(m)}}\left( {v|\hat v, E=0 } \right).
	\end{equation} 
}

The empirical DI is obtained from specific detector. It measures how much uncertainty about the target state has been reduced following the detection of a target by a given detector. Obviously, empirical DI is a positive evaluation indicator, i.e., the larger the empirical DI, the better the performance of the detector.

\theorem[Target Detection Theorem]{ If the DI of the target detection systems  $ \left( \mathbb{V}^{m}, \pi(v^{m}), p(y^{m}|v^{m}), \mathbb{Y}^{m} \right)  $ is  $ I(\boldsymbol{y};v) $, then all the empirical DI $ {I^{(m)}}\left( v;\hat v| E=0 \right) $ less than $ I(\boldsymbol{y};v) $ are achievable. Specifically, if $ m $ is large enough, for any  $\varepsilon  > 0$, there is a detector whose empirical DI satisfies	
	\begin{equation}
		{I^{(m)}}\left( v;\hat v| E=0 \right) >	I(v;\boldsymbol{y}) - 2\varepsilon  
	\end{equation} and the probability of failed decisions is $ P_f^{(m)} \to 0 $. Conversely, there is no detector whose  empirical DI is greater than the DI when $ P_f^{(m)} \to 0 $.	}

\noindent\textit{Proof: See Appendix B.}

The key idea of the proof of the target detection theorem is to determine in which typical set $\mathbb{A}_\varepsilon ^{(m)}\left( {{v}|\boldsymbol{y}} \right) $  the detection sequence $ \hat{v}^{m} $ is contained. For each received sequence $ \boldsymbol{y}^{m} $, the cardinality of the typical set  $\mathbb{A}_\varepsilon ^{(m)}\left( {{v}|\boldsymbol{y}} \right) $  is $ \approx {2^{mH\left( {v|\boldsymbol{y}} \right)}} $.  Then the \textit{a posteriori} entropy of every sequence is $\approx {{\log {2^{mH\left( {v|\boldsymbol{y}} \right)}}} \mathord{\left/
		{\vphantom {{\log {2^{mh\left( {v|\boldsymbol{y}} \right)}}} m}} \right.
		\kern-\nulldelimiterspace} m} = H\left( {v|\boldsymbol{y}} \right) $. The cardinality of  typical set $\mathbb{A}_\varepsilon ^{(m)}\left( {{v}} \right) $ is  $ \approx  {2^{mH\left( v \right)}} $. Then the \textit{a priori} entropy of every symbol is $\approx {{\log {2^{mH\left( {v} \right)}}} \mathord{\left/
		{\vphantom {{\log {2^{mH\left( {v} \right)}}} m}} \right.
		\kern-\nulldelimiterspace} m} = H\left( {v} \right) $.  As a result, the DI of each symbol obtained from $ \boldsymbol{y}^{m} $ is $ H\left( v \right) - H\left( {v|\boldsymbol{y}} \right) = I\left( {v;\boldsymbol{y}} \right) $. In other words, the input typical set has to be divided into subsets of number  $ {2^{m\left[ {H\left( v \right) - H\left( {v|\boldsymbol{y}} \right)} \right]}} = {2^{mI\left( {v;\boldsymbol{y}} \right)}} $ and each subset is marked with $ mI\left( {v;\boldsymbol{y}} \right) $ bit. Therefore, the DI of  $ mI\left( {v;\boldsymbol{y}} \right) $ bit is used to determine a typical set, which is converted to the DI of  $ I\left( {v;\boldsymbol{y}} \right) $ per symbol.

Even though the target detection theorem was developed and proved for a single target, it is generalizable to multiple targets. Besides, the target detection theorem is existential because we provide a rigorous mathematical proof using the central limit theorem and typical sets.

\section{Simulation Results}
In this section, numerical simulations are performed to verify the correctness of the theorems and demonstrate the effectiveness of the SAP detector. For the sake of brevity, we consider the case of a single target in the observation interval. We set the actual position of the target at $ x_{0} = 0 $.

Fig.4 shows $ P_\mathrm{D} $ and $ P_\mathrm{FA}$ versus $ \rm{SNR}$ with different \textit{a priori} probability $ \pi(1) $. It can be seen that the $ P_\mathrm{D} $ and $ P_\mathrm{FA}$ gradually become closer to the $ \pi(1) $ as the $ \rm{SNR} $ gradually approaches zero. Conversely, as the $ \rm{SNR} $ gradually increases, the $P_\mathrm{D} $ and $P_\mathrm{FA} $ converge to 1 and 0, respectively. This is attributed to the radar system becoming more efficient in distinguishing between the true target echo and the background noise as the increase in SNR. As a result, the probability of mistaking noise for a signal is reduced, while the likelihood of detecting the target is increased. The purpose of presenting Fig. 4 is to illustrate the distinction between the DI and the NP criteria. NP criterion is based on the given false alarm probability. In contrast, the DI criterion takes into account given prior probabilities and allows for variations in the false alarm probability.

\begin{figure} 
	\centering 
	\includegraphics[width=3.5in]{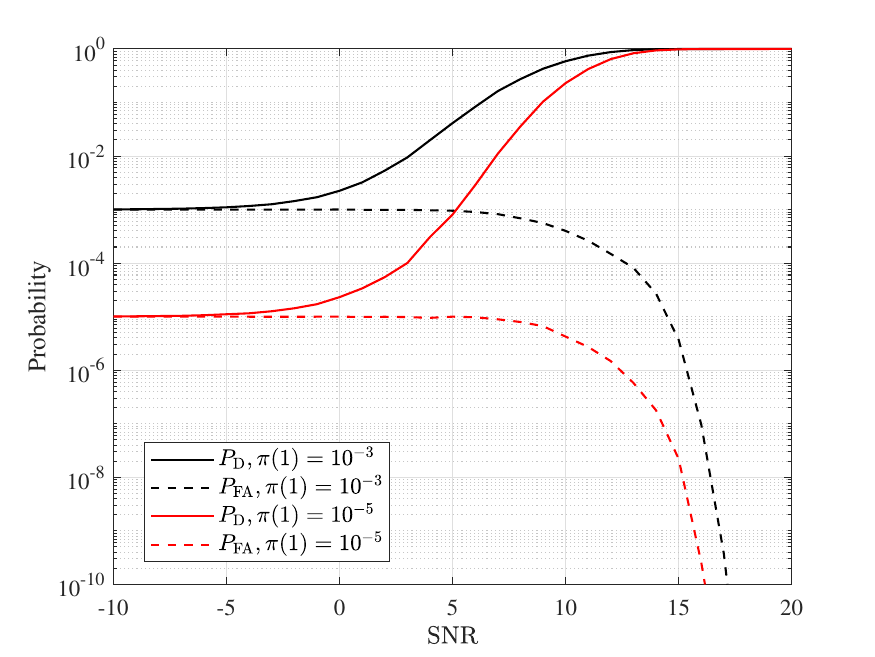} 
	\caption{$ P_\mathrm{FA} $ and $ P_\mathrm{D} $ versus $  \rm{SNR} $ for the observation interval is $ N=32 $ and different  $ \pi(1)$.} 
	\label{fig:side:b} 
\end{figure}

We verify the correctness of the false alarm theorem in Fig. 5 and Fig. 6. It is found that the larger the observation interval is, the better the curves of the false alarm probability and the curves of the \textit{a priori} probability coincide. In Fig. 5 and Fig. 6, the two curves overlap for observation intervals $ N =512 $ and $ N =2048 $, respectively. Therefore, given a sufficiently long observation interval, the false alarm theorem will be valid regardless of the $ \mathrm{SNR}$.

\begin{figure} 
	\centering 
	\includegraphics[width=3.5in]{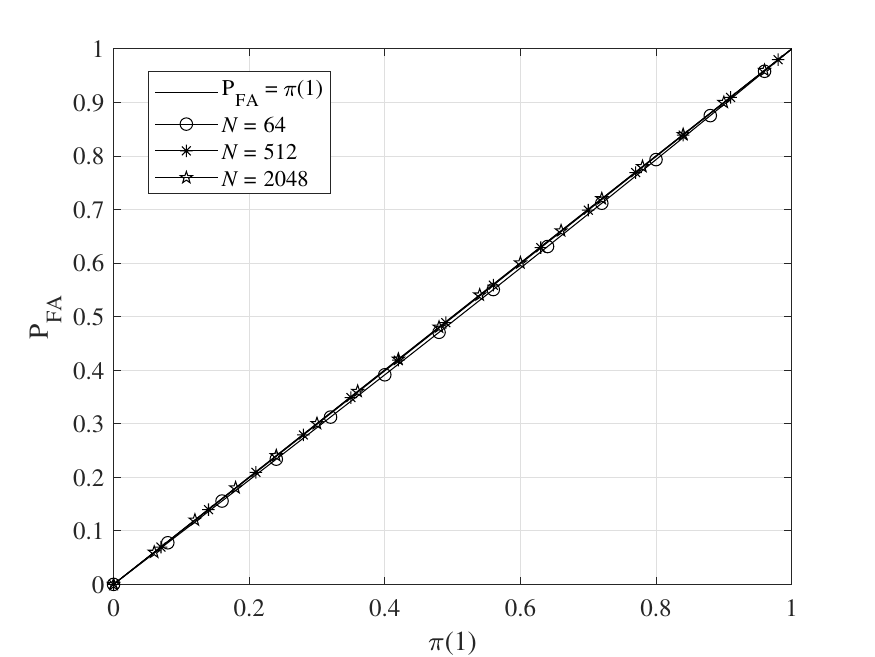} 
	\caption{$ P_\mathrm{FA} $ versus $ \pi(1)$ for the observation interval   $\left( N=64, 512, 2048 \right)$ and $ \mathrm{SNR} =0  \mathrm{dB} $.} 
	\label{fig:side:b} 
\end{figure}

\begin{figure} 
	\centering 
	\includegraphics[width=3.5in]{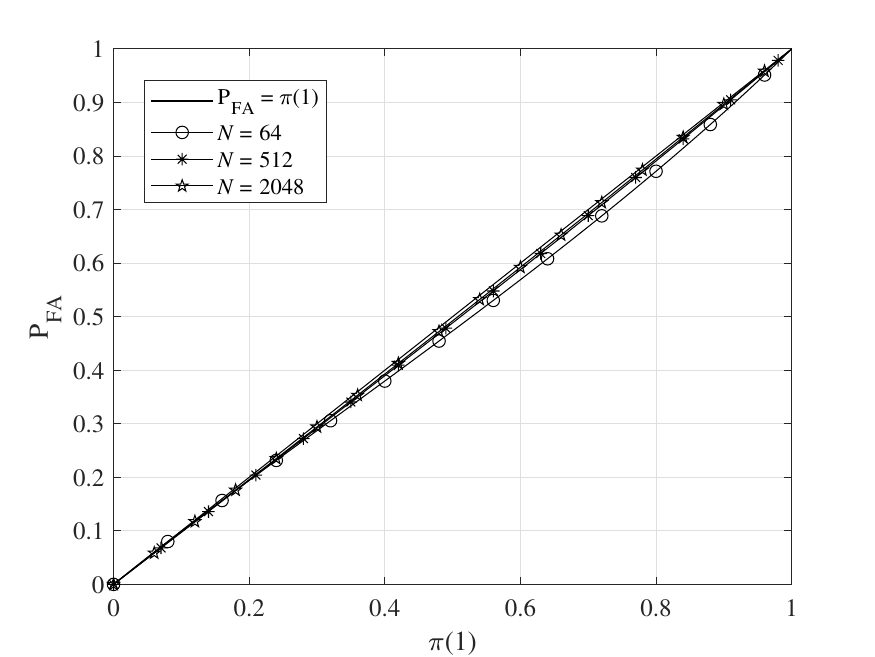} 
	\caption{$ P_\mathrm{FA} $  versus $ \pi(1)$ for the observation interval   $\left( N=64, 512, 2048 \right)$ and $ \mathrm{SNR} =5  \mathrm{dB} $.} 
	\label{fig:side:b} 
\end{figure}

Fig. 7 shows the theoretical performance of SAP detection. It can be seen that the empirical entropy of SAP gradually becomes smoother as the number of snapshots increases. When the number of snapshots is 1000, the empirical entropy of SAP is almost converged. Thus, as long as a sufficient number of snapshots are provided, SAP detection is able to achieve theoretical performance.
\begin{figure} 
	\centering 
	\includegraphics[width=3.5in]{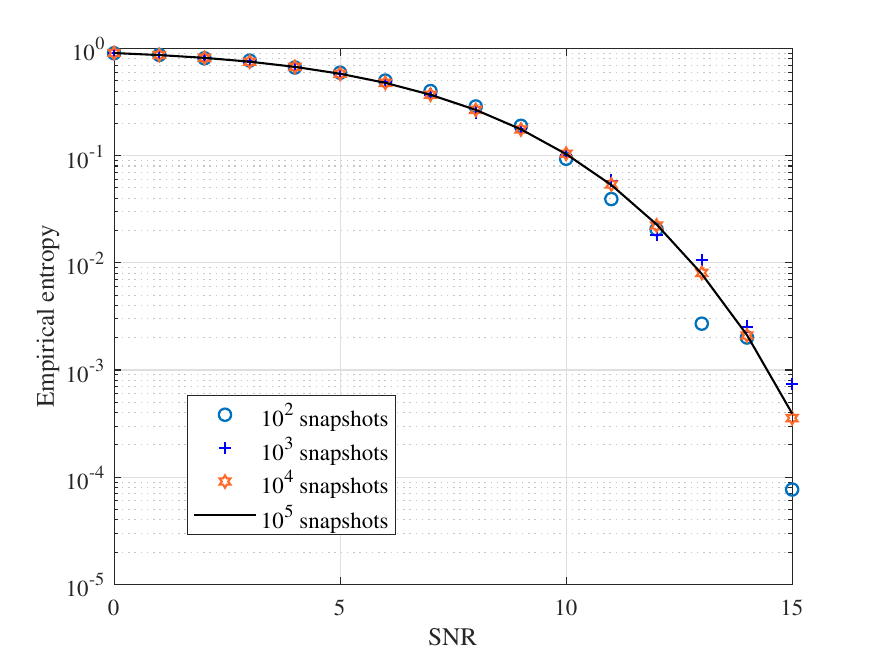} 
	\caption{The empirical entropy of SAP detection versus $ \mathrm{SNR} $ for different snapshots.} 
	\label{fig:side:b} 
\end{figure}

We verify the correctness of the target detection theorem through Fig. 8 and Fig. 9, which show the empirical performance of the MAP and NP detectors and the theoretical performance proposed in this paper. Fig. 8 shows the normalized theoretical DI and the empirical DI versus SNR for the \textit{a priori} probability $ \pi(1)=10^{-3} $. It can be seen that the empirical DI of two detectors is with the theoretical DI as the upper bound. Fig. 9 illustrates the normalized theoretical DI and the empirical DI versus $ \pi(1) $ for $ \mathrm{SNR} $ $ 5 \mathrm{dB} $. As can be seen, the DI equals $ 0 $ when $ \pi \left( 1 \right) $ is $ 0 $ or $ 1 $, meaning there is no uncertainty or additional information available about the target since the target state is known. In contrast, the DI is maximum when $ \pi \left( 1 \right) $ equals $ 0.5 $, as the target state is the most uncertain. It is noted that the curve of the NP detector is the maximum value of empirical DI, which is obtained by searching for all false alarm probabilities.

\begin{figure} 
	\centering 
	\includegraphics[width=3.5in]{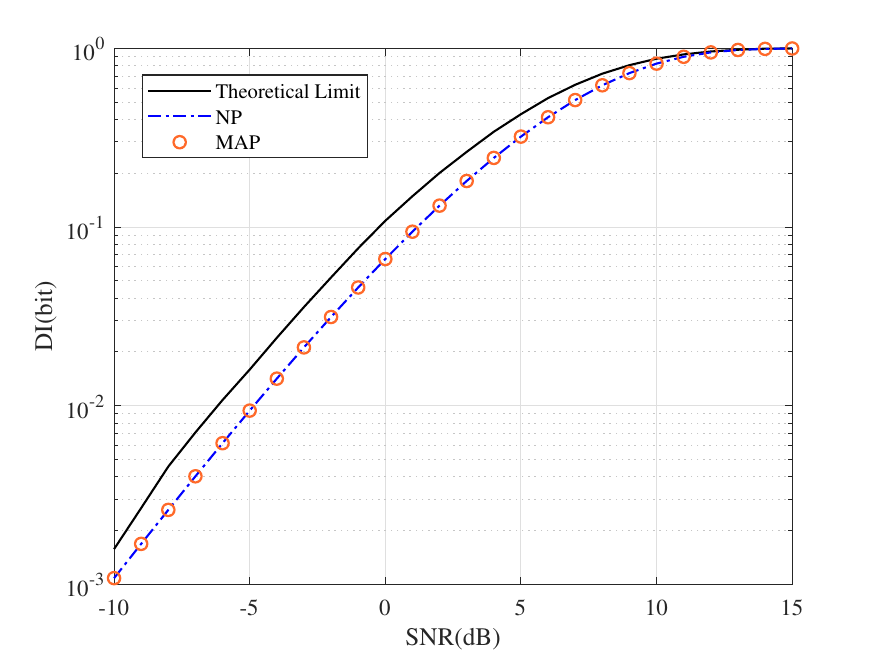} 
	\caption{The DI and the empirical DI versus $ \mathrm{SNR}  $.} 
	\label{fig:side:b} 
\end{figure}

\begin{figure} 
	\centering 
	\includegraphics[width=3.5in]{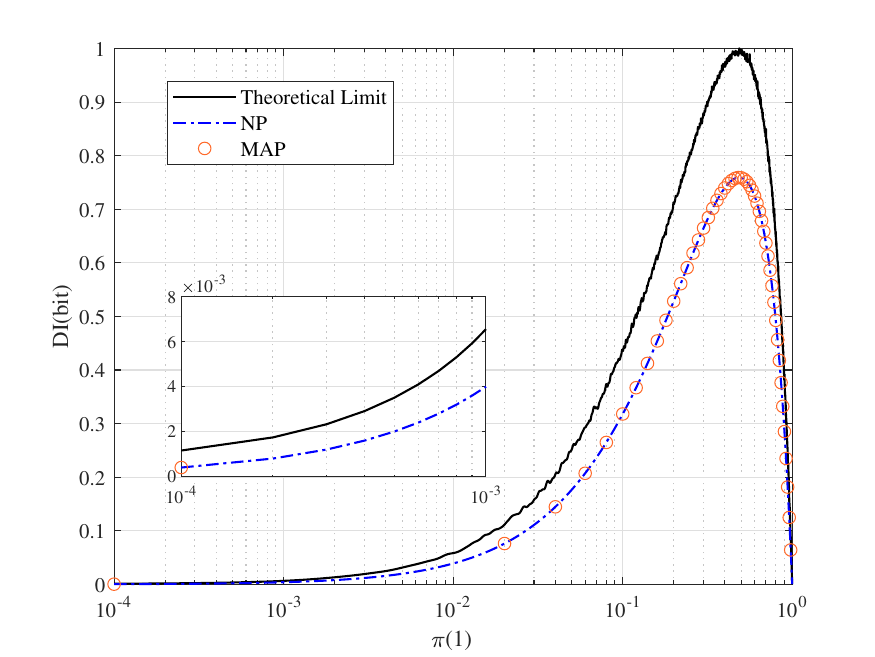} 
	\caption{The normalized theoretical DI and the normalized empirical DI versus  $ \pi \left( 1 \right) $ for $ \mathrm{SNR} = 5 \ \mathrm{dB} $.} 
	\label{fig:side:b} 
\end{figure}

\section{Conclusion}

This paper establishes an information theory framework for target detection employing Shannon's information theory, addressing the quantitative aspect of detection information and laying the information-theoretic foundation for target detection. We demonstrate that detection information is the theoretical limit of target detection, thereby solving the problem of optimal target detection and providing a theoretical basis for various detector designs. Unlike the existence proof of encoding theorems, the proof of the target detection theorem is constructive, indicating that the asymptotically optimal performance of the SAP detector is achievable. The target detection theorem is not only significant in radar signal processing but also has important applications in statistical hypothesis testing problems and mathematical principles of biological population evolution.

\section*{Acknowledgments}
This work was supported by the National Natural Science Foundation of China under Grants 62271254.

\Acknowledgements{This work was supported by National Natural Science Foundation of China (Grant Nos. 00000000 and 11111111).}

\begin{appendix}
\section{Proof of the false alarm theorem}
To ensure the generality of the false alarm theorem, we do not assume a particular scattering model here. Assuming that there is only one target in the observation interval, the received signal is $ {\boldsymbol{y}} = vs{\boldsymbol{u}}{\rm{(}}x{\rm{)}} + {\boldsymbol{w}} $. According to (7), the conditional probability density function of $ \boldsymbol{y} $ is given by 
\begin{equation}
	\begin{split}
		p\left( {{\boldsymbol{y}}\left| v \right.} \right) \!= {\left( {\frac{1}{{\pi {N_0}}}} \right)^N}\exp \left[ { - \frac{1}{{{N_0}}}\left( {{{\boldsymbol{y}}^{\rm{H}}}{\boldsymbol{y}}} \right)} \right]
		\oint \!{\pi (s)\!\exp\! \left( \!{ - \frac{{v{{\left| s \right|}^2}}}{{{N_0}}}} \right)\!\frac{1}{T\!B}\!\!\int_{ - \frac{{T\!B}}{2}}^{\frac{{T\!B}}{2}}\! \!\!{\exp \!\left\{ {\frac{2}{{{N_0}}}{\mathcal{R}}\left[\! {v{s^*}{{\boldsymbol{\psi}}^{\rm{H}}}(x){\boldsymbol{y}}}\! \right]} \right\}\!{\rm{d}}x{\rm{d}}s} }.
	\end{split}	
\end{equation} Then, based on the Bayes formula, the \textit{a posteriori} PDF is given by
	\begin{equation}
	\begin{split}
		P\left( {v\left| {\boldsymbol{y}} \right.} \right) = \frac{1}{\mathcal{Z} }\pi (v)\!
		\oint {\pi (s)\exp \left( { - \frac{{v{{\left| s \right|}^2}}}{{{N_0}}}} \right)\!\frac{1}{T\!B}\!\!\int_{ - \frac{{T\!B}}{2}}^{\frac{{T\!B}}{2}}\! \!\!{\exp \!\left\{ {\frac{2}{{{N_0}}}{\mathcal{R}}\left[\! {v{s^*}{{\boldsymbol{\psi}}^{\rm{H}}}(x){\boldsymbol{y}}}\! \right]} \right\}\!{\rm{d}}x{\rm{d}}s} }. 
	\end{split}
\end{equation}  Equation (A2) can also be expressed as

\begin{equation}
	P\left( {v\left| {\boldsymbol{y}} \right.} \right) = \frac{{\pi \left( v \right)\Upsilon \left( {v,{\boldsymbol{y}}} \right)}}{{\pi \left( 0 \right) + \pi \left( 1 \right)\Upsilon \left( {1,{\boldsymbol{y}}} \right)}},
\end{equation}
where
\begin{equation}
	\begin{aligned}
		\Upsilon \left( {v,{\boldsymbol{y}}} \right) =
		\oint \!\! {\pi (s)\! \exp\!  \left(\!  { - \frac{{v{{\left| s \right|}^2}}}{{{N_0}}}}\!  \right)\!\! \frac{1}{T\!B}\!\! \int_{ - \frac{{T\!B}}{2}}^{\frac{{T\!B}}{2}} {\!\exp \left\{ {\frac{2}{{{N_0}}}{\mathop{\mathcal{R}}\nolimits} \left[ {v{s^ * }{{\boldsymbol{\psi}}^{\rm{H}}}{\rm{(}}x{\rm{)}}{\boldsymbol{y}}} \right]} \!\! \right\}\!\! {\rm{d}}x}{\rm{d}}s} .
	\end{aligned}	
\end{equation}

When the received signal is $ {{\boldsymbol{y}}_0} $, the output of the matched filter is ${{\boldsymbol{\psi}}^{\rm{H}}}{\rm{(}}x{\rm{)}}{\boldsymbol{y} }_0= w{\rm{(}}x{\rm{)}} $. Then, we have
\begin{equation}
	P\left( {1\left| {{\boldsymbol{y}}_0}  \right.} \right) = \frac{{\pi \left( 1 \right)\Upsilon \left( {1,{{\boldsymbol{y}}_0}} \right)}}{{\pi \left( 0 \right) + \pi \left( 1 \right)\Upsilon \left( {1,{{\boldsymbol{y}}_0}} \right)}},
\end{equation} and
\begin{equation}
	\begin{aligned}
	\Upsilon \left( {1,{{\boldsymbol{y}}_0}} \right) = 		
		\oint {\pi (s)\exp \left( { - \frac{{{{\left| s \right|}^2}}}{{{N_0}}}} \right)\frac{1}{T\!B}\int_{ - \frac{{T\!B}}{2}}^{\frac{{T\!B}}{2}} {\exp \left\{ {\frac{2}{{{N_0}}}{\mathop{\mathcal{R}}\nolimits} \left[ {{s^ * }w{\rm{(}}x{\rm{)}}} \right]} \right\}{\rm{d}}x} {\rm{d}}s} .
	\end{aligned}
\end{equation}
The time average of the stationary process is equal to the ensemble average if the observation interval is sufficiently large, then,
\begin{equation}
	\begin{split}
		\Upsilon \left( {1,{{\boldsymbol{y}}_0}} \right) \!=\!
		\displaystyle\! \oint \!\!{\pi (s)\exp \!\left( { - \frac{{{{\left| s \right|}^2}}}{{{N_0}}}} \right)\!\!{{\rm{E}} _w}\!\!\left[ {\exp\! \left\{\! {\frac{2}{{{N_0}}}{\mathop{\mathcal{R}}\nolimits} \left[ {{s^ * }w{\rm{(}}x{\rm{)}}} \right]} \right\}} \right]\!\!{\rm{d}}s} .
	\end{split}
\end{equation}
Because of  $ w \left( t\right) {\sim} {\cal C}\mathcal{N}(\mu ,\sigma^2)$ , we have
\begin{equation}
	{{\rm{E}} _w}\left[ {\exp \left\{ {\frac{2}{{{N_0}}}{\mathop{\rm Re}\nolimits} \left[ {{s^ * }w{\rm{(}}x{\rm{)}}} \right]} \right\}} \right] =\displaystyle \exp \left( {\frac{{{{\left| s \right|}^2}}}{{{N_0}}}} \right).
\end{equation}
Thus
\begin{equation}
	\begin{aligned}
		\Upsilon \left( {1,{{\boldsymbol{y}}_0}} \right) &= \oint {\pi (s)\exp \left( { - \frac{{{{\left| s \right|}^2}}}{{{N_0}}}} \right)\exp \left( {\frac{{{{\left| s \right|}^2}}}{{{N_0}}}} \right){\rm{d}}s} 
		\\&= \oint {\pi (s){\rm{d}}s} 
		\\&= 1.
	\end{aligned}
\end{equation}
Substituting  (A9) into (A5) yields
\begin{equation}
	P\left( {1\left| {{\boldsymbol{y}}_0} \right.} \right) = \frac{{\pi (1)\Upsilon \left( {1,{{\boldsymbol{y}}_0}} \right)}}{{\pi (0) + \pi (1)\Upsilon \left( {1,{{\boldsymbol{y}}_0}} \right)}} = \pi (1).
\end{equation}
Based on the definition of the false alarm probability, one has
\begin{equation}
	{P_{\rm{FA}}} = {{\rm{E}} _{{\boldsymbol{y}}_0}}\left[ {P(1|{{\boldsymbol{y}}_0})} \right].
\end{equation}
According to (A10) and (A11), we have
\begin{equation}
	{P_{\rm{FA}}} =  \pi (1).
\end{equation}

\section{Proof of the target detection theorem}
Fix $\pi \left( v \right) $. Generate the extended target state sequence  $ v^{m} $ according to the distribution,
\begin{equation}
	\pi \left( {{v^m}} \right) = \prod\limits_{i = 1}^m {\pi \left( {{v_i}} \right)}. 
\end{equation} 

The  received sequence $ \boldsymbol{y}^{m} $ is generated after the extended transmitting sequence $ v^{m} $ passed the $ m $ extended channel
\begin{equation}
	p\left( {{{\boldsymbol{y}}^m}\left| {{v^m}} \right.} \right) = \prod\limits_{i = 1}^m {p\left( {{{\boldsymbol{y}}_i}\left| {{v_i}} \right.} \right)} . 
\end{equation} 

Assuming that the detection channel $ p({\boldsymbol{y}}^{m}|{v}^{m}) $ and the \textit{a priori} distribution  $\pi \left( v^{m} \right) $ are known to the receiver, the \textit{a posteriori}  probability is calculated by
\begin{equation}
	P(v|{\boldsymbol{y}}) = \frac{{\pi (v)p({\boldsymbol{y}}|v)}}{{\sum\limits_v {\pi (v)p({\boldsymbol{y}}|v)} }}. 
\end{equation} 

The typical set detection is used to make the detection sequence $ {\hat v^m} $. For the memory-less snapshot channel $\left( {{\mathbb{V}^m},p({{\boldsymbol{y}}^m}|{v^m}),{{{{\rm \mathbb{Y}}}}^m}} \right)$, the extended \textit{a posteriori} PDF of the SAP detection is $ {P_{{{ \rm SAP}}}}\left( {{{\hat v}^m}|{{\boldsymbol{y}}^m}} \right) = P\left( {{{\hat v}^m}|{{\boldsymbol{y}}^m}} \right)$. Then the joint \textit{a posteriori} PDF of the SAP detection satisfies 
\begin{equation}
	\begin{array}{c}
		\begin{aligned}
			{P_{{\rm{SAP}}}}\left( {{{\hat v}^m},{{\boldsymbol{y}}^m}} \right) &= p\left( {{{\boldsymbol{y}}^m}} \right){P_{\rm SAP}}\left( {{{\hat v}^m}|{{\boldsymbol{y}}^m}} \right)\\
			&=p\left( {{{\boldsymbol{y}}^m}} \right)p\left( {{{\hat v}^m}|{{\boldsymbol{y}}^m}} \right)\\
			&= p\left( {{{\hat v}^m},{{\boldsymbol{y}}^m}} \right).
		\end{aligned}
	\end{array} 
\end{equation}  Therefore, $ {\hat v^m} $ and $ {\boldsymbol{y}}^{m} $ are jointly typical sequences and the SAP detection is a method of typical set detection. By Lemma 5. (2), the typical set $ \mathbb{A}_\varepsilon ^m(v|\boldsymbol{y}) $ of satisfies 
\begin{equation}
	\left( {1 - \varepsilon } \right){2^{m\left( {H(v|\boldsymbol{y}) - 2\varepsilon } \right)}} < \left\| {\mathbb{A}_\varepsilon ^m( v|\boldsymbol{y})} \right\| < {2^{m\left( {H(v|\boldsymbol{y}) + 2\varepsilon } \right)}}. 
\end{equation}

If the detection is successful, i.e.,  $ E=0 $, then the empirical entropy satisfies 
\begin{equation}
	{H^{(m)}}(v|\hat v, E=0) < H(v|\boldsymbol{y}) + 2\varepsilon. 
\end{equation} Therefore, the empirical DI satisfies
\begin{equation}
	{I^{(m)}} (v;\hat v| E=0 ) > I(v;\boldsymbol{y}) - 2\varepsilon. 
\end{equation} The achievability of the DI is proved.

There are two events that cause the detection failure for the typical set detection. The first is that $ v^{m} $ and $ \boldsymbol{y}^{m} $ do not form jointly typical sequences, denoted by $ {\bar A_T} $. The second is that $ \hat{v}^{m} $ and $ \boldsymbol{y}^{m} $  do not form jointly typical sequences, denoted by $ {\bar A_R} $. Then, the probability of failed  is 
\begin{equation}
	\begin{array}{c} 
		\begin{aligned}
			P_f^{(m)} &= \Pr \left( {{{\bar A}_T} \cup {{\bar A}_R}} \right)\\
			&\le \Pr \left( {{{\bar A}_T}} \right) + \Pr \left( {{{\bar A}_R}} \right).
		\end{aligned}
		
	\end{array}
\end{equation} According to the lemma 5.(3), we have
\begin{equation}
	P_f^{(m)} \le 2\varepsilon, 
\end{equation} and $P_f^{(m)} $ converges to zero as $ m $ increases.

Below, we prove the converse to the target detection theorem. For this, we will extend Fano's inequality to the target detection. Let us firstly introduce some definitions and a lemma. 

Then we focus on the conditional entropy $ H\left( {{v^m},E|{{{\boldsymbol{y}}}^m}} \right) $. According to the chain rule for entropy, we have
\begin{equation}
	H\left( {{v^m},E|{{\boldsymbol{y}}}^m} \right) = H\left( {E|{{{\boldsymbol{y}}}^m}} \right) + H\left( {{v^m}|{{{\boldsymbol{y}}}^m},E} \right).
\end{equation} It is obvious that $ H\left( {E|{{{\boldsymbol{y}}}^m}} \right) < 1 $. The remaining term $ H\left( {{v^m}|{{{\boldsymbol{y}}}^m},E} \right)$ can be expressed as 
\begin{equation}
	\begin{aligned}
		H\left( {{v^m}|{{{\boldsymbol{y}}}^m},E} \right) &= \left( {1 - P_f^{(m)}} \right)H\left( {{v^m}|{{{\boldsymbol{y}}}^m},E = 0} \right) \\
		&+ P_f^{(m)}H\left( {{v^m}|{{{\boldsymbol{y}}}^m},E = 1} \right),
	\end{aligned}
\end{equation} where $ H\left( {{v^m}|{{{\boldsymbol{y}}}^m},E = 0} \right)  $ denotes the uncertainty when the detection is successful.

According to the property of typical sets, we have
\begin{equation}
	\begin{array}{c}
		\begin{aligned}
			H\left( {{v^m}|{{{\boldsymbol{y}}}^m},E = 0} \right) 
			&\mathop  \le \limits^{\left( a \right)} \log \left\| {\mathbb{A}_\varepsilon ^{(m)}\left( {v|{{\boldsymbol{y}}}} \right)} \right\|\\
			&\mathop  \le \limits^{\left( b \right)} \log {2^{m\left[ {H\left( {v|{{\boldsymbol{y}}}} \right) + 2\varepsilon } \right]}}\\
			&= m\left[ {H\left( {v|{{\boldsymbol{y}}}} \right) + 2\varepsilon } \right],
		\end{aligned}
	\end{array}
\end{equation} where $ \left( a \right) $ is gained by the maximum discrete entropy theorem and $ \left( b \right) $ is resorted to lemma 5.(2). Similarly,
\begin{equation}
	\begin{array}{c}
		\begin{aligned}
			H\left( {{v^m}|{{{\boldsymbol{y}}}^m},E = 1} \right)
			&\le \log \left( {\left\| {\mathbb{A}_\varepsilon ^m(v)} \right\| - \left\| {\mathbb{A}_\varepsilon ^m(v|{{\boldsymbol{y}}})} \right\|} \right)\\
			&\le \log \left( {{2^{m\left[ {H\left( v \right) + \varepsilon } \right]}} - {2^{m\left[ {H\left( {v|{{\boldsymbol{y}}}} \right) - 2\varepsilon } \right]}}} \right)\\
			&\le \log {2^{m\left[ {H\left( v \right) + \varepsilon } \right]}}\\
			&= m\left[ {H\left( v \right) + \varepsilon } \right].
		\end{aligned}
	\end{array}
\end{equation} Therefore, we have the following lemma.

\begin{lemma}\label{Lemma 7: Generalization of Fano's inequality}
	$\textbf{[Extended Fano's Inequality]}$
	\begin{equation}
		\begin{aligned}
			H\left( {{v^m},E|{{{\boldsymbol{y}}}^m}} \right) \!\le  \!1  \!+ \! \left( {1 - P_f^{(m)}} \right)H\left( {{v^m}|{{{\boldsymbol{y}}}^m},E = 0} \right) 
			+ P_f^{(m)}m\left[ {H\left( v \right) + \varepsilon } \right].
		\end{aligned}
	\end{equation}
\end{lemma}
\begin{proof}
	According to (B19) and (B11), we have
	\begin{equation}
		\begin{array}{c}
			\begin{aligned}
				H\left( {{v^m},E|{{{\boldsymbol{y}}}^m}} \right)& = H\left( {E|{{{\boldsymbol{y}}}^m}} \right) + H\left( {{v^m}|E,{{{\boldsymbol{y}}}^m}} \right)\\
				&\le 1 +H\left( {{v^m}|E,{{{\boldsymbol{y}}}^m}} \right)\\
				&= 1 \! + \! \left( { \!1  \!- \!P_f^{(m)}} \! \right) \!H \!\left( {{v^m}|{{{\boldsymbol{y}}}^m} \!,E \! = \! 0} \right) + P_f^{(m)}H\left( {{v^m}|{{{\boldsymbol{y}}}^m},E = 1} \right).
			\end{aligned}
		\end{array}
	\end{equation}
	According to (B12) and (B13), (B15) is rewritten as
	\begin{equation}
		\begin{aligned}
			H\left( {{v^m},E|{{{\boldsymbol{y}}}^m}} \right) \!\le  \!1 \! + \! \left( {1 - P_f^{(m)}} \right)H\left( {{v^m}|{{{\boldsymbol{y}}}^m},E = 0} \right) + P_f^{(m)}m\left[ {H\left( v \right) + \varepsilon } \right].
		\end{aligned}
	\end{equation}
\end{proof}	

Next, the proof of the converse to the target detection theorem	is provided. According to the properties of entropy and mutual information, we have
\begin{equation}
	\begin{array}{c}
		\begin{aligned}
			H\left( {{v^m}} \right) &= H\left( {{v^m}|{{{\boldsymbol{y}}}^m}} \right) + I\left( {{v^m};{{{\boldsymbol{y}}}^m}} \right)\\
			&\le H\left( {{v^m},E|{{{\boldsymbol{y}}}^m}} \right) + I\left( {{v^m};{{{\boldsymbol{y}}}^m}} \right),
		\end{aligned}
	\end{array}
\end{equation} where $ H\left( {{v^m}} \right) = mH\left( v \right) $. By the property of the extended channel, we have
\begin{equation}
	I\left( {{v^m};{{{\boldsymbol{y}}}^m}} \right) \le mI\left( {v;{{\boldsymbol{y}}}} \right).
\end{equation} In light of lemma 6, we have
\begin{equation}
	\begin{aligned}
		mH\left( v \right) \le 1 + \left( {1 - P_f^{(m)}} \right)H\left( {{v^m}|{{{\boldsymbol{y}}}^m},E = 0} \right) 
		+ P_f^{(m)}m\left[ {H\left( v\right) + \varepsilon } \right] + mI\left( {v;{{\boldsymbol{y}}}} \right).
	\end{aligned}
\end{equation} Handling (B19) yields 
	\begin{equation}
	\begin{array}{c}
		\begin{aligned}
			H\left( V \right) - {H^{\left( m \right)}}\left( {V|{\hat V},E=0} \right)& \le \frac{1}{m} - P_f^{(m)}H\left( {{V^m}|{{{Y}}^m},E = 0} \right) + P_f^{(m)}\left[ {H\left( V \right) + \varepsilon } \right] + I\left( {V;{{Y}}} \right)\\
			&\le \frac{1}{m} - P_f^{(m)}\left[ {H\left( {V|{{Y}}} \right) + 2\varepsilon } \right] + P_f^{(m)}\left[ {H\left( V\right) + \varepsilon } \right] + I\left( {V;{{Y}}} \right)\\
			&= \frac{1}{m} + P_f^{(m)}I\left( {V;{{Y}}} \right) - P_f^{(m)}\varepsilon  + I\left( {V;{{Y}}} \right).
		\end{aligned}
	\end{array}
\end{equation} If $ m \to \infty ,P_f^{(m)} \to 0 $, (B20) is expressed as
\begin{equation}
	{I^{(m)}}(v;{\hat v}|E=0 ) < I(v;{{\boldsymbol{y}}}).
\end{equation} 
\end{appendix}

\end{document}